\documentclass[10pt]{article}
\usepackage{graphicx}
\usepackage{amsmath}
\usepackage[per-mode=symbol,detect-weight]{siunitx}
\usepackage{amssymb}
\usepackage{epstopdf}
\usepackage{a4wide}
\usepackage{amsthm}
\newcolumntype{P}[1]{>{\centering\arraybackslash}p{#1}}
\newtheorem{example}{Example}
\usepackage{subfigure}
\usepackage{graphicx,graphics}
\usepackage{hyperref}
\usepackage[vlined,ruled,linesnumbered]{algorithm2e}
\newtheorem{theorem}{Theorem}[section]

\newtheorem{lemma}{Lemma}[section]

\begin{document}
\begin{center}
 {{\bf \large {\rm {\bf  Compact finite difference method for pricing European and American options under jump-diffusion models}}}}
\end{center}
\begin{center}
{\textmd {{\bf Kuldip Singh Patel,}}}\footnote{\it Department of Mathematics, Indian Institute of Technology, Delhi, India, (kuldip@maths.iitd.ac.in)}
{\textmd {{\bf Mani Mehra}}}\footnote{\it Department of Mathematics, Indian Institute of Technology, Delhi, India, {(mmehra@maths.iitd.ac.in)} }
\end{center}
\begin{abstract}
In this article, a compact finite difference method is proposed for pricing European and American options under jump-diffusion models. Partial integro-differential equation and linear complementary problem governing European and American options respectively are discretized using Crank-Nicolson Leap-Frog scheme. In proposed compact finite difference method, the second derivative is approximated by the value of unknowns and their first derivative approximations which allow us to obtain a tri-diagonal system of linear equations for the fully discrete problem. Further, consistency and stability for the fully discrete problem are also proved. Since jump-diffusion models do not have smooth initial conditions, the smoothing operators are employed to ensure fourth-order convergence rate. Numerical illustrations for pricing European and American options under Merton jump-diffusion model are presented to validate the theoretical results.
\end{abstract}
\begin{center}
\end{center}
\section{Introduction}\label{sec:intro}
\par F. Black and M. Scholes \cite{BlaS73} derived a partial differential equation (PDE) governing the option prices in the stock market with the assumption that the dynamics of underlying asset is driven by geometric Brownian motion with constant volatility. Though Black-Scholes model is a seminal work in option pricing, numerous studies found that these assumptions are inconsistent with the market price movements. Therefore, various approaches have been considered to overcome the shortcomings of Black-Scholes model. In one of these approaches, Merton \cite{Mer76} incorporated the jumps into the dynamics of underlying asset in order to determine the volatility skews and it is known as Merton jump-diffusion model. In another approach, S. L. Heston \cite{Hes93} considered the volatility to be a stochastic process and this model is known as stochastic volatility model. Apart from these, Dupire \cite{Dup94} considered the volatility to be a deterministic function of time and stock price. Further, Bates \cite{Bates96} combined the jump-diffusion model with stochastic volatility approach to capture the typical features of market option prices. Anderson and Andreasen \cite{Ander00} combined the deterministic volatility function approach with jump-diffusion model and proposed a second-order accurate numerical method for valuation of options.
\par The prices of European options under jump-diffusion models can be evaluated by solving a partial integro-differential equation (PIDE), whereas a linear complementary problem (LCP) is solved for the evaluation of American options. Let us introduce some existing literature on numerical methods for the solution of the PIDE and LCP. Cont and Voltchkova \cite{Cont05} used implicit-explicit (IMEX) scheme for pricing European and barrier options and proved the stability and convergence of the proposed scheme. d'Halluin et al. \cite{Halluin05} proposed a second-order accurate implicit method for pricing European and American options which uses fast Fourier transform (FFT) for the evaluation of convolution integral. They also proved the stability and the convergence of the fixed-point iteration method. An excellent comparison of various approaches for jump-diffusion models is given in \cite{Duffy05}. A three-time levels second-order accurate implicit method using finite difference approximations is proposed for European and American put options under jump-diffusion models in \cite{Kwon11} and \cite{Kwonamer11} respectively. Salmi et al. \cite{SaToSy14} proposed a second-order accurate IMEX time semi-discretization scheme for pricing European and American options under Bates model. They explicitly treated the jump term using the second-order Adams-Bashforth method and rest of the terms are discretized implicitly using the Crank-Nicolson method.
\par It is observed that the inclusion of more grid points in computation stencil in order to increase the accuracy of finite difference approximations becomes computationally expensive. Therefore, finite difference approximations have been developed using compact stencils (commonly known as compact finite difference approximations) at the expense of some complication in their evaluation. Compact finite difference approximations provide high-order accuracy and better resolution characteristics as compared to finite difference approximations for equal number of grid points \cite{Lele92}. A detailed study about various order compact approximations is presented in \cite{MKPACM17}. Compact finite difference approximations have also been used for option pricing problems \cite{TanGB08, HWSun11}.
\par The majority of numerical approaches \cite{Halluin05, Kwon11, Kwonamer11, SaToSy14} proposed to price European and American options under jump-diffusion models are based on second-order discretization methods. Nevertheless, high-order approximations are not customary tools for option pricing because initial conditions for option pricing are always non-smooth. As a result, it will affect the convergence rate of high-order methods. Various approaches e.g. co-ordinate transformation \cite{TanGB08} and local mesh refinements \cite{HWSun11} have been considered for option pricing problems to achieve high-order convergence rate even for non-smooth initial conditions. These approaches suffer with certain drawbacks e.g. it is not always easy to define a coordinate transformation for PIDE and the stability results for using local mesh refinement are not straight forward. Therefore as another approach, we apply smoothing operator to the initial conditions to obtain high-order convergence rate even for non-smooth initial conditions \cite{thomee70}.
\par In this article, a compact finite difference method is proposed to solve the PIDE and LCP for pricing European and American options under jump-diffusion models. The novelty of the proposed compact finite difference method is that it does not require the original equation as an auxiliary equation unlike the compact scheme proposed in \cite{HWSun11}. The consistency and stability of the proposed compact finite difference method are proved. Since initial conditions for jump-diffusion models have low regularity, the smoothing operator given in \cite{thomee70} is employed to smoothen the initial conditions in order to achieve the high-order convergence rate. Further, the CPU times for a given accuracy with proposed compact finite difference method and finite difference method are calculated and it is shown that proposed compact finite difference method outperforms the finite difference method.
\par The outline of the paper is as follows. The continuous model problem is discussed in Sec.~\ref{sec:conti_prob}. In Sec.~\ref{sec:compact}, compact finite difference approximations for first and second derivatives are discussed and Fourier analysis of errors is presented. In Sec.~\ref{sec:disc_prob}, compact finite difference method for pricing European and American options is proposed. Consistence and stability analysis for European options is discussed in Sec.~\ref{sec:analysis}. In Sec.~\ref{sec:numerical}, numerical examples are presented to validate the theoretical results. Finally, conclusions and some future work are discussed in Sec.~\ref{sec:conclu}.
\section{The mathematical model}
\label{sec:conti_prob}
A brief discussion on continuous problems for pricing European and American options under jump-diffusion models is presented in this section. Let us consider that stock price process of an underlying asset follows an exponential jump-diffusion model, i.e.
\[
S_{t}=S_{0}e^{rt+X_{t}},
\]
where $S_{0}$ is the stock price at $t=0$, $r$ is the risk-free interest rate and $(X_{t})_{t\geq 0}$ a jump-diffusion L$\acute{e}$vy process \cite{Kwon11}. The jump-diffusion L$\acute{e}$vy process $(X_{t})_{t\geq 0}$ is defined as
\begin{equation}
\label{eq:jdprocess}
X_{t}:=at+\sigma W_{t}+\sum_{i=1}^{N_{t}} G_{i},
\end{equation}
where  $a$ and $\sigma > 0$ are real constants, $(W_{t})_{t\geq 0}$ is Brownian motion, $(N_{t})_{t \geq 0}$ is Poisson process and $G_{i}$ are identically and independent distributed random variables. Further, the random variables $G_{i}$ follows Gaussian distribution in case of Merton jump-diffusion model. The price of European options under jump-diffusion models ($V(S,t)$) is obtained by solving a PIDE which is discussed in the following theorem \cite{Kwon11}.
\begin{theorem} Let the L$\acute{e}$vy process $\left(X_{t}\right)_{t\geq 0}$ has the L$\acute{e}$vy triplet $(\sigma^2, \gamma, \nu)$, where $\sigma >0$, $\gamma \in \mathbb{R}$
	and $\nu$ is the L$\acute{e}$vy measure. If
	\[
	\sigma>0 \:\: or \;\: \exists \:\: \beta \in (0,2)\:\:\:\: \mbox{such that} \:\:\: \liminf \limits_{\epsilon\rightarrow 0} \epsilon^{-\beta}\int_{-\epsilon}^{\epsilon} |x|^2 \nu(dx) >0,
	\]
	then the value of European option with the payoff function $Z(S_{T})$ is obtained by $V(S,t)$, where
	\[
	V:[0,\infty)\times[0,T] \rightarrow \mathbb{R},
	\]
	\[
	(S,t)\mapsto V(S,t)=e^{-r(T-t)}\mathbb{E}[Z(S_{T})|S_{t}=S],
	\]
	is a continuous map on $[0,\infty)\times[0,T]$, $C^{1,2}$ on $(0,\infty)\times(0,T)$, and satisfies the following PIDE
	\begin{equation}
	\begin{split}
	\label{eq:PIDE}
	-\frac{\partial V}{\partial t}(S,t)&=\frac{\sigma^2 S^2}{2}\frac{\partial^2 V}{\partial S^2}(S,t)+rS\frac{\partial V}{\partial S}(S,t)-rV(S,t)\\
	&+\int_{\mathbb{R}}^{}\left[ V(Se^{x},t)-V(S,t)-S(e^x-1)\frac{\partial V}{\partial S}(S,t)\right]\nu(dx),
	\end{split}
	\end{equation}
	on $(0,\infty)\times[0,T)$ with the final condition
	\[
	V(S,T)=Z(S) \:\:\:\: \forall \:\: S>0.
	\]
\end{theorem}
Let us consider the following transformation in the above PIDE (\ref{eq:PIDE})
\[
\tau=T-t, \:x=ln\left(\frac{S}{S_{0}}\right) \:\mbox{and} \:u(x,\tau) = V(S_{0}e^{x},T-\tau).
\]
Then, $u(x,\tau)$ is the solution of the following PIDE with constant coefficients
\begin{equation}
\begin{split}
\label{eq:pidefinal}
\frac{\partial u}{\partial \tau}(x,\tau)&=\mathbb{L}u,\: (x,\tau)\in \:(-\infty,\infty)\times(0,T],\\
u(x,0) &= f(x)\:\:\: \forall \:\:\: x \in (-\infty, \infty),
\end{split}
\end{equation}
where
\small
\begin{equation}
\label{eq:operator}
\mathbb{L}u=\frac{\sigma^2}{2}\frac{\partial^2 u}{\partial x^2}(x,\tau)+\left(r-\frac{\sigma^2}{2}-\lambda \zeta\right)\frac{\partial u}{\partial x}(x,\tau)-(r+\lambda)u(x,\tau)+\lambda \int_{\mathbb{R}}^{} u(y,\tau)g(y-x)dy,
\end{equation}
\normalsize
$\lambda$ is the intensity of the jump sizes and $\zeta$ = $\int_{\mathbb{R}}^{} (e^x-1)g(x)dx$.
\par Further, the LCP for American options is written as
\[
\frac{\partial u}{\partial \tau}(x,\tau)-\mathbb{L}u(x,\tau) \geq 0,
\]
\begin{equation}
\label{eq:american}
u(x,\tau)\geq f(x),
\end{equation}
\[
\left(\frac{\partial u}{\partial \tau}(x,\tau)-\mathbb{L}u(x,\tau)\right) \left(u\left(x,\tau\right)-f(x)\right)= 0,
\]
for all $(x,\tau)\in (-\infty,\infty)\times(0,T]$ and $\mathbb{L}$ is given in Eq.~(\ref{eq:operator}). The initial condition for European put options is
\begin{equation}
\label{eq:initial_eur}
f(x)= max(K-S_{0}e^{x},0) \:\:\: \forall \:\:\:x \in \mathbb{R},
\end{equation}
and the asymptotic behaviour of European put options is described as
\begin{equation}
\label{eq:boundary_eur}
\lim_{x\rightarrow -\infty}[u(x,\tau)-(Ke^{-r\tau}-S_{0}e^{x})]=0\:\:\:\:\:\mbox{and} \:\:\:\:\:\lim_{x\rightarrow \infty}u(x,\tau)=0.
\end{equation}
Similarly, the initial condition for American put options is
\begin{equation}
\label{eq:initial_amer}
f(x)= max(K-S_{0}e^{x},0) \:\:\: \forall \:\:\:x \in \mathbb{R},
\end{equation}
and the equations describing the asymptotic behaviour of European call options are
\begin{equation}
\label{eq:boundary_amer}
\lim_{x\rightarrow -\infty}[u(x,\tau)-(K-S_{0}e^{x})]=0\:\:\:\:\:\mbox{and} \:\:\:\:\:\lim_{x\rightarrow \infty}u(x,\tau)=0.
\end{equation}
\section{Compact finite difference approximations for first and second derivatives}
\label{sec:compact}
Let us consider the fourth-order compact finite difference approximations for first and second derivatives \cite{Lele92} of function $u$ as follows
\begin{align}
\frac{1}{4}u_{x_{i-1}}+u_{x_{i}}+\frac{1}{4}u_{x_{i+1}} &= \frac{1}{\delta x}\left[-\frac{3}{4}u_{i-1}+\frac{3}{4}u_{i+1}\right]
\label{eq:firstc_4},\\
\frac{1}{10}u_{xx_{i-1}}+u_{xx_{i}}+\frac{1}{10}u_{xx_{i+1}} &= \frac{1}{\delta x^2}\left[\frac{6}{5}u_{i-1}-\frac{12}{5}u_{i}+\frac{6}{5}u_{i+1}\right]\label{eq:pade2_4},
\end{align}
where $u_{x_{i}}$ and $u_{xx_{i}}$ represents first and second derivatives approximations of unknown $u$ at grid point $x_{i}$. If $\Delta_{x}u_{i}$ and $\Delta_{xx}u_{i}$ represent second-order finite difference approximation for first and second derivative respectively, then we may write
\begin{equation}
\label{eq:firstf_2}
\Delta_{x}u_{i}=\frac{u_{i+1}-u_{i-1}}{2\delta x},\:\:\:\Delta_{xx}u_{i}=\frac{u_{i+1}-2u_{i}+u_{i-1}}{\delta x^2}.
\end{equation}
If the first derivative of unknowns are also considered as variables then Eq.~(\ref{eq:firstc_4}) can be written as
\begin{equation}
\label{eq:pade3_4}
\frac{1}{4}u_{xx_{i-1}}+u_{xx_{i}}+\frac{1}{4}u_{xx_{i+1}}= \frac{1}{\delta x}\left[-\frac{3}{4}u_{x_{i-1}}+\frac{3}{4}u_{x_{i+1}}\right].
\end{equation}
Eliminating $u_{xx_{i-1}}$ and $u_{xx_{i+1}}$ from Eqs.~(\ref{eq:pade2_4}) and~(\ref{eq:pade3_4}) and using Eq.~(\ref{eq:firstf_2}) we have
\begin{equation}
\label{eq:secondc_4}
u_{xx_{i}} = 2\Delta_{xx}u_{i}-\Delta_{x}u_{x_{i}}.
\end{equation}
In this way, compact finite difference approximation for second derivative is expressed in terms of the value of the functions and their first derivative approximations. The value of $u_{x_{i}}$ in Eq.~(\ref{eq:secondc_4}) is obtained from Eq.~(\ref{eq:firstc_4}). In case of non-periodic boundary conditions, fourth order accurate one-sided compact finite difference approximation for first derivative at boundary point can be obtained from \cite{Tian11}. It can be observed from Fig.~\ref{fig:grid} that lesser number of grid points are needed to achieve high-order accuracy as compared to finite difference approximations.
\begin{figure}[h!]
  \begin{center}
  \includegraphics[trim = 0cm 20cm 0cm 3cm, clip, width=1\textwidth]{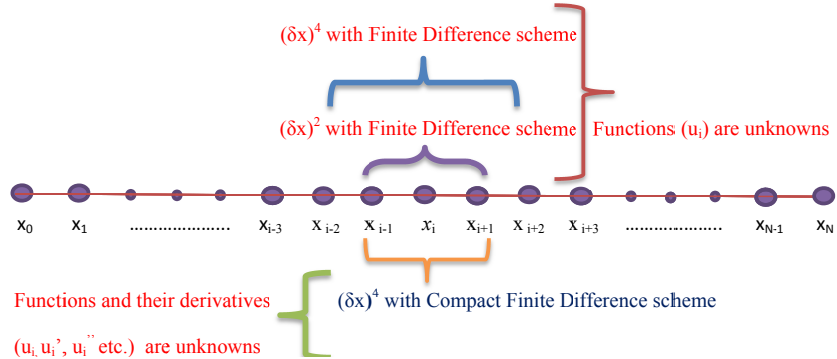}
  \caption{Number of grid points required for first derivative approximation using compact finite difference approximation and finite difference approximation.}
  \label{fig:grid}
  \end{center}
\end{figure}
\par The emphasis in this section now is on the resolution characteristics of the compact finite difference approximations of first and second derivatives rather than its truncation error. Fourier analysis is used to obtain the dispersion and dissipation errors which quantify the resolution characteristics of difference approximations. In order to discuss the Fourier analysis of errors, the dependent variable $(u(x))$ is assumed to be periodic over the domain $[0,L]$ of the independent variable $(x)$ and write
\begin{equation}\label{eq:fourier}
u(x)=\sum_{k=-N/2}^{N/2}\hat{u}_{k}e^{\frac{2\pi I k x}{L}},
\end{equation}
where $I=\sqrt{-1}$. Now, the Fourier modes are defined as $e^{I\omega s}$ where $\omega=\frac{2\pi k\delta x}{L}$ is the wavenumber, $N$ is the number of grid points and $s=\frac{x}{\delta x}$ is the scaled coordinate. The exact first and second derivatives of Eq.~(\ref{eq:fourier}) provide functions with Fourier coefficients $\hat{u}'_{k}=I \omega \hat{u}_{k}$ and $\hat{u}''_{k}= -\omega^2 \hat{u}_{k}$ respectively. The differencing errors are obtained by comparing the Fourier coefficients of the exact derivatives with the Fourier coefficients of first and second derivative approximations. If $\omega'$ and $\omega''$ represent the modified wavenumbers for first and second derivatives respectively then following relations \cite{KPMD17} are obtained for various difference approximations:
\small
\begin{equation}\label{eq:integer1}
\omega' = \left\{
\begin{array}{lr}
sin(\omega) & : \mbox{$O(\delta x^2)$ finite difference approximation},\\
\frac{-sin\left(2\omega\right)}{6}+\frac{4sin(\omega)}{3} & : \mbox{$O(\delta x^4)$ finite difference approximation},\\
\frac{3sin(\omega)}{2+cos(\omega)} & : \mbox{$O(\delta x^4)$ compact finite difference approximation}.
\end{array}
\right.
\end{equation}
\begin{equation}\label{eq:integer2}
\omega'' = \left\{
\begin{array}{lr}
2-2cos(\omega) & : \mbox{$O(\delta x^2)$ finite difference approximation},\\
\frac{cos(2\omega)}{6}-\frac{8cos(\omega)}{3}+\frac{5}{2} & : \mbox{$O(\delta x^4)$ finite difference approximation}.\\
\frac{12(1-cos(\omega))}{2+cos(\omega)} & : \mbox{$O(\delta x^4)$ compact finite difference approximation(Eq.~(\ref{eq:pade2_4}))},\\
\frac{5-4cos(\omega)-cos^{2}(\omega)}{2+cos(\omega)} & : \mbox{$O(\delta x^4)$ compact finite difference approximation (Eq.~(\ref{eq:secondc_4}))}.
\end{array}
\right.
\end{equation}
\normalsize
\begin{figure}[h!]
	\begin{center}
		\subfigure[]{%
			\includegraphics[scale=0.45]{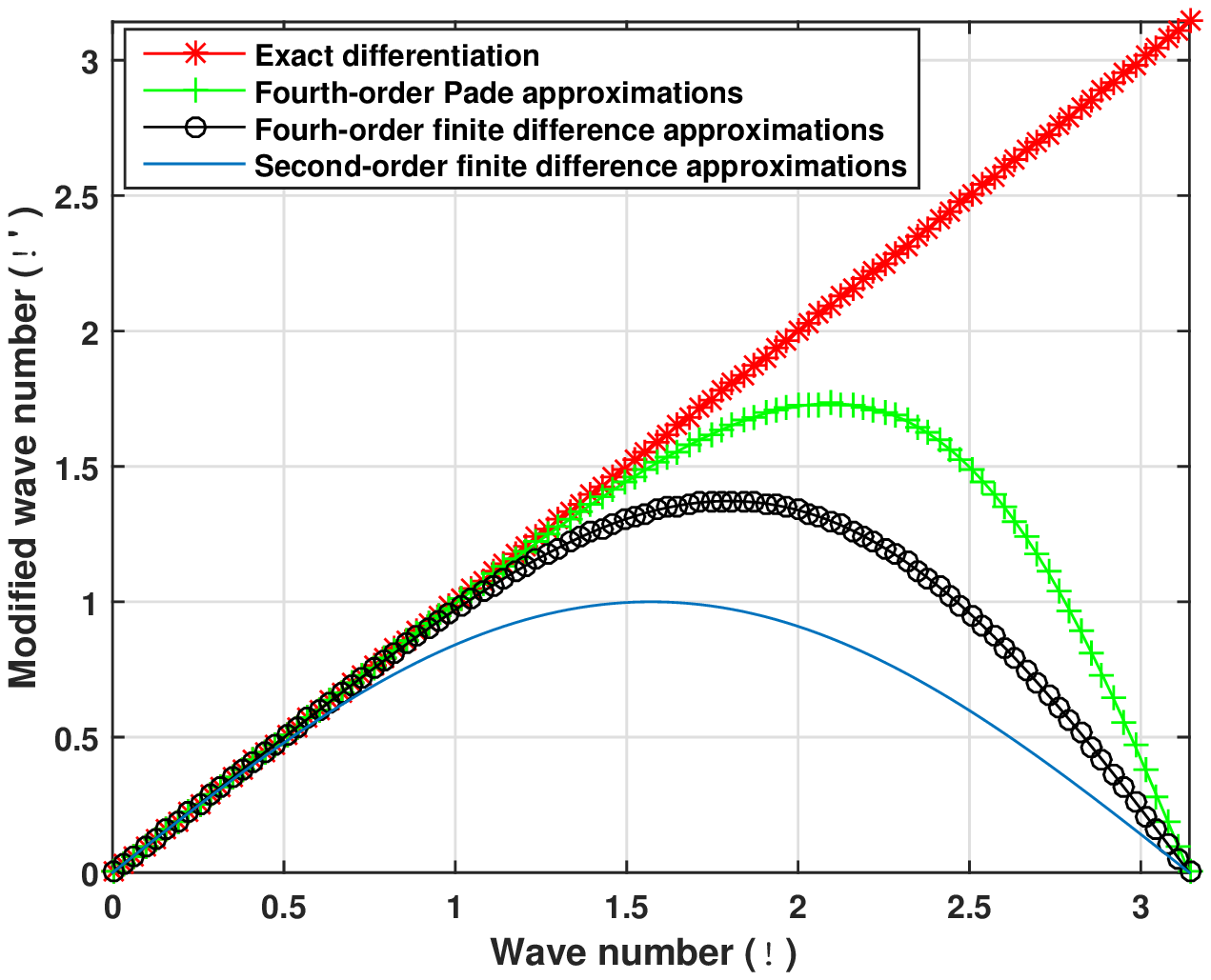}
			\label{fig:fou_first}}%
		\subfigure[]{%
			\includegraphics[scale=0.45]{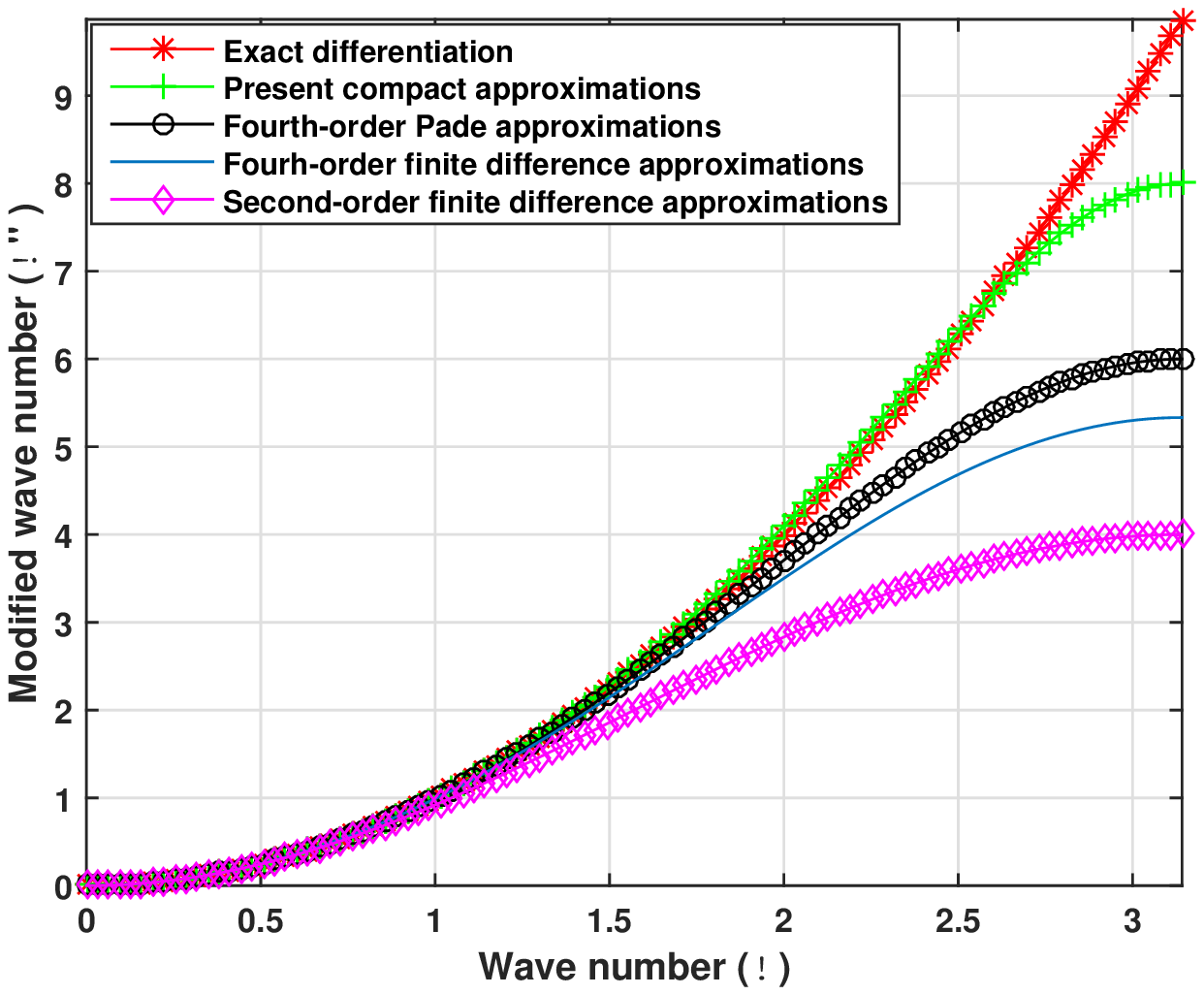}
			\label{fig:fou_second}}
		\caption{The wave number and modified wave number for various finite difference approximations: (a) First derivative approximation, (b). Second derivative approximation.}
		\label{fig:fou}
	\end{center}
\end{figure}
\par The differences between real parts of the wave number and modified wave numbers $(Re(\omega'-\omega), Re(\omega''-\omega))$ and the imaginary parts of the wavenumber and modified wavenumbers $(Im(\omega'-\omega), Im(\omega''-\omega))$ represent the dispersion and dissipation errors respectively. Since all the difference approximations discussed in Sec.\ref{sec:compact} are of the central difference form, there are no dissipation errors involved. The wave numbers versus modified wave numbers for first and second derivatives (given in Eqs.~(\ref{eq:integer1}) and~(\ref{eq:integer2})) are plotted in Figures~(\ref{fig:fou_first}) and~(\ref{fig:fou_second}) respectively. It is observed that fourth-order compact finite difference approximations have lesser dispersion error as compared to the finite difference approximations. Moreover, we observe that the proposed compact finite difference approximation for second derivative $(\ref{eq:secondc_4})$ has better resolution characteristics as compared to Pad$\acute{e}$ approximation ($\ref{eq:pade2_4}$).
\section{Compact finite difference method}
\label{sec:disc_prob}
\subsection{Localization to bounded domain}
\label{ssec:localization}
The domain of the spatial variable is restricted to a bounded interval $\Omega=(-L,L)$ for some fixed real number $L$ to solve the PIDE (\ref{eq:pidefinal}) numerically. For given positive integers $M$ and $N$, let $\delta x=2L/N$
and $\delta \tau=T/M$ and in this way we define $x_{n}=-L+n\delta x$ $(n=0,1,....,N)$
and $\tau_{m}=m\delta \tau$ $(m=0,1,...,M)$. Cont and Voltchkova \cite{Cont05} proved that truncation error after localization decreases exponentially point-wise. Further, Matache et al. \cite{Mat05} also proved an exponential bound in $L_{2}$-norm on truncation error. Now, the PIDE~(\ref{eq:pidefinal}) can be written as
\begin{equation}
\label{eq:discre_1}
\frac{\partial u(x,\tau)}{\partial \tau}=\mathbb{D}u(x,\tau)+\mathbb{I}u(x,\tau),\:\:\: (x,\tau) \in \Omega \times [0,T),
\end{equation}
where $\mathbb{D}$ corresponds to the differential operator and $\mathbb{I}$ represents the integral operator. The operators $\mathbb{D}$ and $\mathbb{I}$ are as follows
\begin{equation}
\label{eq:all_operator}
\begin{split}
\mathbb{D}u(x,\tau) &= \frac{\sigma^2}{2}\frac{\partial^2 u}{\partial x^2}(x,\tau)+\left(r-\frac{\sigma^2}{2}-\lambda \zeta\right)\frac{\partial u}{\partial x}(x,\tau)-(r+\lambda)u(x,\tau),\\
\mathbb{I}u(x,\tau) &= \lambda \int_{\mathbb{R}}^{}u(y,\tau)g(y-x)dy.
\end{split}
\end{equation}
\subsection{Temporal semi-discretization}
\label{ssec:temp_disc}
\par Crank-Nicolson Leap-Frog scheme is used for time semi-discretization of Eq.~(\ref{eq:discre_1}) as follows:
\begin{equation}
\label{eq:temp_discre1}
\frac{u^{m+1}-u^{m-1}}{2\delta\tau}=\mathbb{D}\left(\frac{u^{m+1}+u^{m-1}}{2}\right)+\lambda \mathbb{I}(u^{m}),\:\:\: m \geq 1,
\end{equation}
\begin{equation*}
\label{eq:temp_discre2}
u^{m+1}(x_{min})=K(e^{-r\tau_{m+1}}-e^{x_{min}}),\:\:\: u^{m+1}(x_{max})=0.
\end{equation*}
It is known that for all $v(.,\tau) \in \mathbb{L}^2(\Omega)$, $\tau\in(0,T)$ which is defined as
\begin{equation*}
v(x,\tau) = \left\{
  \begin{array}{lr}
    u(x,\tau) & : (x,\tau)\in \Omega \times [0,T],\\
    0 & : (x,\tau)\in \Omega^c \times [0,T],
  \end{array}
\right.
\end{equation*}
the integral operator satisfies the following condition
\begin{equation}
\label{eq:integral}
||\mathbb{I}v(.,\tau)||\leq C_1||u(.,\tau)||,
\end{equation}
where $C_1$ is a constant independent of $\tau$. Let us suppose $u^{m}$ and $\tilde{u}^{m}$ represents the exact and approximate solution of the Eq.~(\ref{eq:temp_discre1}) respectively with error $e^{m}:=u^{m}-\tilde{u}^{m}$. In order to prove the stability for temporal semi-discretization, following theorem is proved.
\begin{theorem} There exist a constant $\gamma$ such that $\forall$ $\delta \tau < \frac{1}{\gamma}$, we have
\begin{equation}
||e^{i}||^2\leq C||e^0||^2,\:\:\: \forall\:\: 2 \leq i \leq M,
\end{equation}
where $C$ is a constant depends on $r$, $\sigma$, $\lambda$, $T$ and $C_1$.
\end{theorem}
\begin{proof}
The error equation for temporal semi-discretization can be written as
\begin{equation}
\label{eq:temp_discre5}
\frac{e^{m+1}-e^{m-1}}{2\delta\tau}=\mathbb{D}\left(\frac{e^{m+1}+e^{m-1}}{2}\right)+\lambda \mathbb{I}(e^{m}),
\end{equation}
\begin{equation*}
\label{eq:temp_discre6}
e^{m+1}(x_{min})=0,\:\:\: e^{m+1}(x_{max})=0.
\end{equation*}
Taking the inner product with $(e^{m+1}+e^{m-1})$ in above Eq.~(\ref{eq:temp_discre5}), we get
\begin{equation}
\label{eq:temp_discre7}
\left(\frac{e^{m+1}-e^{m-1}}{2\delta\tau},(e^{m+1}+e^{m-1})\right)=\mathbb{D}\left(\frac{e^{m+1}+e^{m-1}}{2},(e^{m+1}+e^{m-1})\right)+\lambda \left[\mathbb{I}(e^{m}),(e^{m+1}+e^{m-1})\right],
\end{equation}
which implies
\begin{equation}
\label{eq:temp_discre7}
\begin{split}
\frac{||e^{m+1}||^2-||e^{m-1}||^2}{2\delta\tau}&=-\frac{\sigma^2}{2}\frac{||e_{x}^{m+1}||^2-||e_{x}^{m-1}||^2}{2}+\left(r-\frac{\sigma^2}{2}-\lambda\zeta\right)\left(\frac{e_x^{m+1}+e_{x}^{m-1},e^{m+1}+e^{m-1}}{2}\right)\\
&-(r+\lambda)\frac{||e^{m+1}+e^{m-1}||^2}{2}+\lambda\left[\mathbb{I}(e^{m}),e^{m+1}+e^{m-1}\right].
\end{split}
\end{equation}
After simplification, we have
\begin{equation}
\begin{split}
\label{eq:temp_discre8}
\frac{||e^{m+1}||^2-||e^{m-1}||^2}{2\delta\tau}&=-\frac{\sigma^2}{4}\left(||(e_{x}^{m+1}+e_{x}^{m-1})-\frac{(r-\frac{\sigma^2}{2}-\lambda\zeta)}{\sigma^2}(e^{m+1}+e^{m-1})||^2\right)\\
&+Z||e^{m+1}+e^{m-1}||^2+\lambda\left[\mathbb{I}(e^{m}),e^{m+1}+e^{m-1}\right],
\end{split}
\end{equation}
where $Z=\frac{(r-\frac{\sigma^2}{2}-\lambda \zeta)^2-2(r+\lambda)\sigma^2}{4\sigma^2}$. Now using Eq.~(\ref{eq:integral}) and applying triangle inequality, we get
\begin{equation}
\begin{split}
\label{eq:temp_discre9}
\frac{||e^{m+1}||^2-||e^{m-1}||^2}{2\delta\tau}&\leq 2Z\left(||e^{m+1}+e^{m-1}||^2\right)+\frac{\lambda}{2}\left(C_1||e^{m}||^2+2||e^{m+1}||^2+2||e^{m-1}||^2\right),\\
&\leq \frac{1}{2}\left(\gamma_1||e^{m+1}||^2+\gamma_2||e^{m}||^2+\gamma_3||e^{m-1}||^2\right),\\
&\leq \frac{\gamma}{6}\left(||e^{m+1}||^2+||e^{m}||^2+||e^{m-1}||^2\right),\\
\end{split}
\end{equation}
where $\gamma_1=\gamma_3=4Z+2\lambda$, $\gamma_2=C_1 \lambda$ and $\gamma=6 \left(max \{\gamma_1,\gamma_2,\gamma_3\}\right)$. Without loss of generality suppose that $i$ is an even number and adding up Eq.~(\ref{eq:temp_discre8}) for odd $m$ between $1$ to $i-1$, we obtain
\begin{equation}
\begin{split}
\label{eq:temp_discre10}
\frac{||e^{i}||^2-||e^{0}||^2}{\delta\tau}&\leq \frac{\gamma}{3}\left(\sum_{m=2,\:m:even}^{i}||e^{m}||^2+\sum_{m=1,\:m:odd}^{i-1}||e^{m}||^2+\sum_{m=0,\:m:even}^{i-2}||e^{m}||^2\right),\\
&\leq \frac{\gamma}{3}\left(\sum_{m=2}^{i}||e^{m}||^2+\sum_{m=1}^{i-1}||e^{m}||^2+\sum_{m=0}^{i-2}||e^{m-1}||^2\right),\\
&\leq \gamma\sum_{m=0}^{i}||e^{m}||^2.\\
\end{split}
\end{equation}
Rearranging the terms in the above Eq.~(\ref{eq:temp_discre10}), we get
\begin{equation}
\label{eq:temp_discre11}
||e^{i}||^2\leq \delta \tau\gamma\sum_{m=0}^{i}||e^{m}||^2+||e^{0}||^2.\\
\end{equation}
Now, applying discrete Gronwall's inequality \cite{KadT15}, we get the desired result.
\end{proof}
\subsection{The Fully Discrete Problem}
\label{ssec:three_time}
\par The numerical approximations for the differential operator $\mathbb{D}$ and the integral operator $\mathbb{I}$ are discussed in this section. If $\mathbb{D}_{\delta}$ represents the discrete approximations for the operator $\mathbb{D}$, then
\begin{equation}
\label{eq:full_discre1}
\mathbb{D_{\delta}}{u^{m}_{n}}=\frac{\sigma^2}{2}u^{m}_{xx_{n}}+\left(r-\frac{\sigma^2}{2}-\lambda \zeta\right
)u^{m}_{x_{n}}-(r+\lambda)u^{m}_{n},
\end{equation}
where $u^{m}_{n}=u(x_{n}\tau_{m})$ and $u^{m}_{x_{n}}$, $u^{m}_{xx_{n}}$ are the first and second derivative approximations of $u(x_{n},\tau_{m})$ respectively. Now using Eq.~(\ref{eq:secondc_4}) in above Eq.~(\ref{eq:full_discre1}), we get
\begin{equation}
\label{eq:ddelta}
\mathbb{D_{\delta}}{u^{m}_{n}}=\frac{\sigma^2}{2}\left(2\Delta^2_{x}u^{m}_{n}-\Delta_{x}u_{x_{n}}^{m}\right)+\left(r-\frac{\sigma^2}{2}-\lambda \zeta\right)u^{m}_{x_{n}}-(r+\lambda)u^{m}_{n}.
\end{equation}
In this way, second derivative approximation of unknowns are eliminated from the PIDE using the unknowns itself and their first derivative approximation.
\par Now, the discrete approximation for the integral operator $\mathbb{I}u$ using fourth-order accurate composite Simpson's rule is discussed. Integral operator $\mathbb{I}u(x,\tau)$ given in Equation~(\ref{eq:all_operator}) is divided into two parts namely on $\Omega=(-L,L)$ and $\mathbb{R}\backslash\Omega$. If $\Upsilon(x,\tau,L)$ represents the value of integral operator $\mathbb{I}u$ on $\mathbb{R}\backslash\Omega$, then
\begin{equation}
\Upsilon(x,\tau,L) = \left\{
  \begin{array}{lr}
    Ke^{-r\tau}\Phi\left(-\frac{x+\mu_{J}+L}{\sigma_{J}}\right)-S_{0}e^{x+\frac{\sigma^2_{J}}{2}+\mu_{J}}\Phi\left(-\frac{x+\sigma^2_{J}+\mu_{J}+L}{\sigma_{J}}\right), & \mbox{(European put options)},\\
    K\Phi\left(-\frac{x+\mu_{J}+L}{\sigma_{J}}\right)-S_{0}e^{x+\frac{\sigma^2_{J}}{2}+\mu_{J}}\Phi\left(-\frac{x+\sigma^2_{J}+\mu_{J}+L}{\sigma_{J}}\right), & \mbox{(American put options)},
  \end{array}
\right.
\end{equation}
where $\Phi(y)$ is the cumulative distribution function of standard normal random variable. The value of integral $\mathbb{I}u(x,\tau)$ on the interval $\Omega$ using composite Simpson's rule is given as
\small
\begin{equation}
\begin{split}
\label{eq:simp}
\int_{\Omega}^{}u(y,\tau_{m})g(y-x_{n})dy &= \frac{\delta x}{3}\left(u^{m}_{0}g_{n,0}+
4\sum_{i=1}^{\frac{N}{2}} u^{m}_{2i-1}g_{n,2i-1}+2\sum_{i=1}^{\frac{N}{2}-1} u^{m}_{2i}g_{n,2i}+u^{m}_{N}g_{n,N}\right)\\
&+O(\delta x^4),
\end{split}
\end{equation}
\normalsize
where $g_{n,i}=g(x_{i}-x_{n})$. In order to write the above integral approximation~(\ref{eq:simp}) in matrix-vector multiplication form, we define
\begin{center}
	$B_{g}=\frac{\delta x}{3}\left[
	\begin{array}{ccccc}
	4g(x_1-x_1) & 2g(x_2-x_1) & 4g(x_3-x_1) &  \dots  & 4g(x_{N-1}-x_1) \\
	4g(x_1-x_2) & 2g(x_2-x_2) & 4g(x_3-x_2) &  \dots  & 4g(x_{N-1}-x_2) \\
	4g(x_1-x_3) & 2g(x_2-x_3) & 4g(x_3-x_3) &  \dots  & 4g(x_{N-1}-x_3) \\
	\dots &   \dots & \dots & \dots & \dots \\
	4g(x_1-x_{N-1}) & 2g(x_2-x_{N-1}) & 4g(x_3-x_{N-1}) &  \dots  & 4g(x_{N-1}-x_{N-1}) \\
	\end{array}
	\right],$
\end{center}
\begin{equation*}
u^{m}=\begin{bmatrix}
u_{1}^{m} \\         u_{2}^{m} \\   \vdots \\     u_{N-1}^{m}
\end{bmatrix},\:\:\:
P^{m}=\frac{\delta x}{3}\begin{bmatrix}
u_{0}^{m}g_{n,0} \\         0 \\     \vdots \\    u_{N}^{m}g_{n,N}
\end{bmatrix}.
\end{equation*}
The matrix $B_{g}$ can be transformed into a Toeplitz matrix by transferring the coefficient $[4,2,4,...,2,4]^{T}$ to the vector $u$ as follows
\begin{center}
	$\tilde{B}_{g}=\frac{\delta x}{3}\left[
	\begin{array}{ccccc}
	g(x_1-x_1) & g(x_2-x_1) & g(x_3-x_1) &  \dots  & g(x_{N-1}-x_1) \\
	g(x_1-x_2) & g(x_2-x_2) & g(x_3-x_2) &  \dots  & g(x_{N-1}-x_2) \\
	g(x_1-x_3) & g(x_2-x_3) & g(x_3-x_3) &  \dots  & g(x_{N-1}-x_3) \\
	\dots &   \dots & \dots & \dots & \dots \\
	g(x_1-x_{N-1}) & g(x_2-x_{N-1}) & g(x_3-x_{N-1}) &  \dots  & g(x_{N-1}-x_{N-1}) \\
	\end{array}
	\right],$
\end{center}
and
\[
\tilde{u}^{m}=\left[4u_{1}^{m},2u_{2}^{m},..,2u_{N-2}^{m},4u_{N-1}^{m}\right]^T.
\]
The above matrix-vector product $(\tilde{B}_{g}\tilde{u}^{m})$ is obtained with $O(N\log{}N)$ complexity by embedding the matrix $\hat{B}_{g}$ in a circulant matrix and using FFT for matrix-vector multiplication \cite{Chan96, Chan07}. Therefore, the discrete approximation $(\mathbb{I}_{\delta}u)$ for the integral operator $(\mathbb{I}u)$ is
\begin{equation}
\mathbb{I_{\delta}}u^{m}=\lambda\left(\tilde{B}_{g}\tilde{u}^{m}+P^{m}+\Upsilon(x,\tau,L)\right).
\end{equation}
If $\mathbb{L}_{\delta}$ denote the discrete approximation of operator $\mathbb{L}$ (defined in Equation~(\ref{eq:operator})), then
\begin{equation}
\label{eq:ldelta}
\mathbb{L}_{\delta}u^m_{n}=\mathbb{D_{\delta}}\left(\frac{u^{m+1}_{n}+u^{m-1}_{n}}{2}\right)+\mathbb{I}_{\delta}u^m_{n}.
\end{equation}
We find $U^{m}_{n}$ (the approximate value of $u^{m}_{n}$) which is the solution of following problem
\begin{equation}
\label{eq:pidediscre}
\frac{U^{m+1}_{n}-U^{m-1}_{n}}{2\delta\tau}=\mathbb{D_{\delta}}\left(\frac{U^{m+1}_{n}+U^{m-1}_{n}}{2}\right)+\mathbb{I_{\delta}}U^{m}_{n},\:\:\:\mbox{$1\leq m\leq M-1$,  $1\leq n\leq N-1$},
\end{equation}
Using the values of $D_{\delta}U^{m}_{n}$ from Eq.~(\ref{eq:ddelta}) in Eq.~(\ref{eq:pidediscre}), we obtain
\begin{equation}
\label{eq:corec2}
\begin{split}
\frac{U^{m+1}_{n}-U^{m-1}_{n}}{2\delta\tau}
&=\frac{1}{2}\left[\frac{\sigma^2}{2}\left(2\Delta_{x}^2U^{m+1}_{n}-\Delta_{x}U_{x_{n}}^{m+1}\right)+\left(r-\frac{\sigma^2}{2}-\lambda \zeta\right)U_{x_{n}}^{m+1}-(r+\lambda)U^{m+1}_{n}\right]\\
&+\frac{1}{2}\left[\frac{\sigma^2}{2}\left(2\Delta_{x}^2U^{m-1}_{n}-\Delta_{x}U_{x_{n}}^{m-1}\right)+\left(r-\frac{\sigma^2}{2}-\lambda \zeta\right)U_{x_{n}}^{m-1}-(r+\lambda)U^{m-1}_{n}\right]\\
&+\mathbb{I_{\delta}}U^{m}_{n},\:\:\:\mbox{for $1\leq m\leq M-1$}.
\end{split}
\end{equation}
Re-arranging the terms in above equation, the following fully discrete problem is obtained
\begin{equation}
\label{eq:corec3}
\begin{split}
\left[1-\delta \tau \frac{\sigma^2}{2} 2 \Delta_{x}^2-\delta \tau(r+\lambda)\right]U^{m+1}_{n}&={\delta \tau}\left[-\frac{\sigma^2}{2}\Delta_{x}U_{x_{n}}^{m+1}+\left(r-\frac{\sigma^2}{2}-\lambda \zeta\right)U_{x_{n}}^{m+1}\right]\\
&+\delta\tau\left[\frac{\sigma^2}{2}(2\Delta_{x}^2U^{m-1}_{n}-\Delta_{x}U_{x_{n}}^{m-1})+\left(r-\frac{\sigma^2}{2}-\lambda \zeta\right)U_{x_{n}}^{m-1}\right.\\
&\left. -(r+\lambda)U^{m+1}_{n} \right]+U^{m-1}_{n}+2\delta \tau \mathbb{I_{\delta}}U^{m}_{n},\:\:\:\mbox{for $1 \leq m\leq M-1$}.
\end{split}
\end{equation}
Let us introduce the following notation
\[
\textbf{U}^m=(U_{1}^m,U_{2}^m,...,U_{N-1}^m)^T \: \mbox{and} \:\:\:\textbf{U}_{x}^m=(U_{x_{1}}^m,U_{x_{2}}^m,...,U_{x_{N-1}}^m)^T,
\]
the resulting system of equations corresponding to the difference scheme (\ref{eq:corec3}) can be written as
\begin{equation}
\label{eq:corec4}
A\textbf{U}^{m+1}=F(\textbf{U}^{m}, \textbf{U}^{m-1}, \textbf{U}_{x}^{m-1},\textbf{U}_{x}^{m+1}).
\end{equation}
The presence of $\textbf{U}_{x}^{m+1}$ on the right hand side of the Equation~(\ref{eq:corec4}) bind us to use a predictor corrector method to solve the system of equations. Therefore, correcting to convergence approach \cite{Lambert91} is used and also summarized in the following algorithm.\\
\textbf{Algorithm for Correcting to Convergence Approach}\\
1. Start with $\textbf{U}^{m}$. \\
2. Obtain $\textbf{U}_{x}^{m}$ using Equation~(\ref{eq:firstc_4}). \\
3. Take $\textbf{U}^{m+1}_{old}=\textbf{U}^{m}$, $\textbf{U}_{x_{old}}^{m+1}=\textbf{U}_{x}^{m}$.\\
4. Correct to $\textbf{U}^{m+1}_{new}$ using Equation~(\ref{eq:corec3}).\\
5. If $\|\textbf{U}^{m+1}_{new}-\textbf{U}^{m+1}_{old}\|_{\infty}$ $<$ $\epsilon$, then $\textbf{U}^{m+1}_{new}=\textbf{U}^{m+1}_{old}$.\\
6. Obtain $\textbf{U}_{x_{new}}^{m+1}$ using Equation~(\ref{eq:firstc_4}).\\
7. Take $\textbf{U}^{m+1}_{old}=\textbf{U}^{m+1}_{new}$, $\textbf{U}^{m+1}_{x_{old}}=\textbf{U}^{m+1}_{x_{new}}$ and go to step $4$.\\
The stopping criterion for inner iteration can be set at $\epsilon=10^{-12}$ in above approach. Since the proposed compact scheme~(\ref{eq:corec4}) is three-time levels, two initial values on the zeroth and first time levels are required to start the computation. The initial condition provides the value of $u$ at $\tau =0$ and the value of $u$ at first time level is obtained by IMEX-scheme used in \cite{Cont05}.
\par In above discussed approach, number of iterations  to achieve desired accuracy are not known in advance. Let number of iterations required by above approach be $n_{m}$ at a fixed time level $m$ and $n_{s}$:=$ \displaystyle \max_{1\leq m \leq M}n_{m}$. We know that a tri-diagonal system of equations is solved with $O(N)$ operations and we have also discussed that matrix-vector multiplication is obtained with $O(NlogN)$ complexity. Therefore, maximum computational complexity of the proposed compact finite difference method will be of order $O\left((n_{s}+log N)NM\right)$.
\par Now, the fully discrete problem for American options using compact finite difference method is discussed. Ikonen et. al. \cite{Iko04} proposed the operator splitting technique for American put options under Black-Scholes model and it is extended by Toivanen \cite{Toi08} for jump-diffusion models. For detailed explanation about the operator splitting technique, one can see \cite{Kwonamer11}. A new auxiliary variable $\psi$ is taken such that $\psi=U_{\tau}-\mathbb{L}U$ and LCP~(\ref{eq:american}) is written as follows:
\begin{align}
\label{eq:lcp}
\begin{split}
U_{\tau}-\mathbb{L}U=\psi,\\
\psi \geq 0,\:\:\:\:U \geq f,\:\:\:\: \psi(U-f)=0.
\end{split}
\end{align}
The above equation is discretized using operator splitting technique as follows:
\begin{equation}
\label{eq:hat1}
\frac{U_{n}^{m+1}-U_{n}^{m-1}}{2\delta \tau}-\left[\mathbb{D}_{\delta}\left(\frac{U_{n}^{m+1}+U_{n}^{m-1}}{2}\right)+\mathbb{I}_{\delta}U_{n}^{m}\right]=\Psi_{n}^{m},
\end{equation}
\begin{equation}
\label{eq:hat2}
\frac{U_{n}^{m+1}-U_{n}^{m-1}}{2\delta \tau}-\left[\mathbb{D}_{\delta}\left(\frac{U_{n}^{m+1}+U_{n}^{m-1}}{2}\right)+\mathbb{I}_{\delta}U_{n}^{m}\right]=\Psi_{n}^{m+1}.
\end{equation}
Now, a pair $(\Psi_{n}^{m+1}, U_{n}^{m+1})$ is to be obtained satisfying the Eqs.~(\ref{eq:hat1}) and~(\ref{eq:hat2}) and the constraints
\begin{equation}
U_{n}^{m+1} \geq f(x_{n}), \:\:\:\:\:\: \Psi_{n}^{m+1} \geq 0, \:\:\:\:\:\: \Psi_{n}^{m+1}\left(U_{n}^{m+1}-f(x_{n})\right)=0.
\end{equation}
An algorithm to solve the above equations is presented in Algorithm~\ref{alg:algorithm}. The system of linear equations obtained from Algorithm~\ref{alg:algorithm} are solved using the correcting to convergence approach which has already been discussed.
\begin{algorithm}[h!]
	%\rule{\textwidth}{1pt}
	\begin{tabbing}
		\hspace{12pt} \= \hspace{14pt} \= \hspace{14pt} \= \kill
		\> {\bf for} $m=0$ \\
		\>\>  {\bf for} $n=1,2,...,N-1$\\
		\> \>\>{\bf $\frac{U_{n}^{m+1}-U_{n}^{m}}{\delta \tau}=\mathbb{D}_{\delta}U_{n}^{m+1}+\mathbb{I}_{\delta}U_{n}^{m}+\Psi_{n}^{m}$} \\
		\>\>{\bf end}\\
		\>\>{\bf Solve for $n=1,2,...,N-1$}\\
		\>\>\> {\bf $U_{n}^{m+1}=max\left(f(x_{n}),U^{m+1}_{n}-\delta \tau \Psi^{m}_{n}\right)$}\\
		\>\>\> {\bf $\Psi^{m+1}_{n}=\frac{U^{m+1}_{n}-U^{m+1}_{n}}{\delta \tau}+\Psi^{m}_{n}$}\\
		\>{\bf end}\\
		
		\> {\bf for} $m \geq 1$ \\
		\>\>  {\bf for} $n=1,2,...,N-1$\\
		\>\>\> {\bf $\frac{U_{n}^{m+1}-U_{n}^{m-1}}{2\delta\tau}=\mathbb{D}_{\delta}\left(\frac{U_{n}^{m+1}
				+U_{n}^{m-1}}{2}\right)+\mathbb{I}_{\delta}U_{n}^{m}+\Psi_{n}^{m}$}\\
		\>\>{\bf end}\\
		\>\>{\bf Solve for $n=1,2,...,N-1$}\\
		\>\>\> {\bf $U_{n}^{m+1}=max\left(f(x_{n}),U^{m+1}_{n}-2\delta \tau \Psi^{m}_{n}\right)$}\\
		\>\>\>{\bf $\Psi^{m+1}_{n}=\frac{U^{m+1}_{n}-U^{m+1}_{n}}{2\delta \tau}+\Psi^{m}_{n}$}\\
		\>{\bf end}\\
	\end{tabbing}
	\caption{Algorithm for evaluating American options.\label{alg:algorithm}}
\end{algorithm}
\section{Consistency and stability analysis}
\label{sec:analysis}
\subsection{Consistency}
\label{ssec:consistency}
\par The consistency of the proposed compact finite difference method~(\ref{eq:corec3}) is proved in the following theorem.
\begin{theorem}
For sufficiently small $\delta x$ and $\delta \tau$, we have
\begin{equation}
\label{eq:temp_cons}
\frac{\partial u}{\partial \tau}(x_n,\tau_m)-\mathbb{L}u(x_n,\tau_m)-\left(\frac{u(x_n,\tau_{m+1})-u(x_n,\tau_{m-1})}{2 \delta \tau}-\mathbb{L}_{\delta}u(x_n,\tau_m)\right)=O(\delta \tau^2+\delta x^4)\:\:\:\:\: \mbox{for $m\geq 1$},
\end{equation}
where $\mathbb{L}$ and $\mathbb{L}_{\delta}$ are given in Eqs.~(\ref{eq:operator}) and~(\ref{eq:ldelta}) respectively and $(x_n,\tau_m) \in (-L,L)\times(0,T]$.
\end{theorem}
\begin{proof} The second-order accurate finite difference approximation for time derivative $\left(\frac{\partial u}{\partial \tau}\right)$ using Taylor series expansion is obtained as follows
\begin{equation}
\label{eq:timeerr}
\left|\frac{\partial u}{\partial \tau}(x_n,\tau_m)-\frac{u(x_n,\tau_{m+1})-u(x_n,\tau_{m-1})}{2 \delta \tau}\right|\\
 \leq \frac{\delta \tau^2}{6} \sup_{\tau \in [\tau_{m-1}, \tau_{m+1}]}\left|\frac{\partial^3 u}{\partial \tau^3}(x_n,\tau)\right|.
 \end{equation}
Further, Taylor series expansion for second derivative provides
\begin{equation*}
\frac{1}{2}\left[\frac{\partial^2 u}{\partial x^2}(x_{n},\tau_{m+1})+\frac{\partial^2 u}{\partial x^2}(x_{n},\tau_{m-1})\right]=\frac{\partial^2 u}{\partial x^2}(x_{n},\tau_{m})+\frac{\delta \tau^2}{2}\frac{\partial^4 u}{\partial x^2 \partial \tau^2}(x_{n},\tau_{m})+O(\delta \tau^3).
\end{equation*}
From above equation, we have
 \[
 \left \lvert \frac{\partial^2 u}{\partial x^2}(x_n,\tau_m)-\frac{1}{2}\left[\frac{\partial^2 u}{\partial x^2}(x_n,\tau_{m+1})+\frac{\partial^2 u}{\partial x^2}(x_n,\tau_{m-1})\right]\right\rvert \leq \frac{\delta \tau^2}{2} \sup_{\tau \in [\tau_{m-1}, \tau_{m+1}]}\left|\frac{\partial^4 u}{\partial x^2 \partial \tau^2 }(x_n,\tau)\right|.
\]
Since the compact finite difference approximations (discussed in Sec.~\ref{sec:compact}) are fourth-order accurate, we can write
\[
\left\lvert\frac{\partial^2 u}{\partial x^2}(x_n,\tau_{m+1})-u_{xx_{n}}^{m+1}\right\rvert=O(\delta x^4),\:\:\:\:\:\left \lvert\frac{\partial^2 u}{\partial x^2}(x_n,\tau_{m-1})-u_{xx_{n}}^{m-1}\right\rvert=O(\delta x^4).
\]
Therefore
\begin{equation*}
\begin{aligned}
\frac{\partial^2 u}{\partial x^2}(x_n,\tau_m)-\frac{1}{2}\left[u_{xx_{n}}^{m+1}+u_{xx_{n}}^{m-1}\right]
& = \frac{\partial^2 u}{\partial x^2}(x_n,\tau_m)-\frac{1}{2}\left[u_{xx_{n}}^{m+1}+u_{xx_{n}}^{m-1}\right]-\frac{1}{2}\frac{\partial^2 u}{\partial x^2}(x_n,\tau_{m+1})\\
& + \frac{1}{2}\frac{\partial^2 u}{\partial x^2}(x_n,\tau_{m+1})-\frac{1}{2}\frac{\partial^2 u}{\partial x^2}(x_n,\tau_{m-1})+\frac{1}{2}\frac{\partial^2 u}{\partial x^2}(x_n,\tau_{m-1}),\\
& = \frac{\partial^2 u}{\partial x^2}(x_n,\tau_m)-\frac{1}{2}\left[\frac{\partial^2 u}{\partial x^2}(x_n,\tau_{m+1})+\frac{\partial^2 u}{\partial x^2}(x_n,\tau_{m-1})\right],\\
& + \frac{1}{2}\left[\frac{\partial^2 u}{\partial x^2}(x_n,\tau_{m+1})-u_{xx_{n}}^{m+1}\right]+\frac{1}{2}\left[\frac{\partial^2 u}{\partial x^2}(x_n,\tau_{m-1})-u_{xx_{n}}^{m-1}\right]\\
& = O(\delta \tau^2+\delta x^4).\\
\end{aligned}
\end{equation*}
Similarly, for first derivative approximation we have
\[
 \left\lvert\frac{\partial u}{\partial x}(x_n,\tau_m)-\frac{1}{2}\left[\frac{\partial u}{\partial x}(x_n,\tau_{m+1})+\frac{\partial u}{\partial x}(x_n,\tau_{m-1})\right]\right\rvert \leq \frac{\delta \tau^2}{2} \sup_{\tau \in [\tau_{m-1}, \tau_{m+1}]}\left|\frac{\partial^3 u}{\partial x \partial \tau^2 }(x_n,\tau)\right|,
\]
\[
\left\lvert\frac{\partial u}{\partial x}(x_n,\tau_{m+1})-u_{x_{n}}^{m+1}\right\rvert=O(\delta x^4),\:\:\:\:\:\left\lvert\frac{\partial u}{\partial x}(x_n,\tau_{m-1})-u_{x_{n}}^{m-1}\right\rvert=O(\delta x^4).
\]
Thus, we get
\begin{equation*}
\begin{split}
\frac{\partial u}{\partial x}(x_n,\tau_m)-\frac{1}{2}\left[u_{x_{n}}^{m+1}+u_{x_{n}}^{m-1}\right]
& = \frac{\partial u}{\partial x}(x_n,\tau_m)-\frac{1}{2}\left[u_{x_{n}}^{m+1}+u_{x_{n}}^{m-1}\right]-\frac{1}{2}\frac{\partial u}{\partial x}(x_n,\tau_{m+1})\\
& + \frac{1}{2}\frac{\partial u}{\partial x}(x_n,\tau_{m+1})-\frac{1}{2}\frac{\partial u}{\partial x}(x_n,\tau_{m-1})+\frac{1}{2}\frac{\partial u}{\partial x}(x_n,\tau_{m-1})\\
& = \frac{\partial u}{\partial x}(x_n,\tau_m)-\frac{1}{2}\left[\frac{\partial u}{\partial x}(x_n,\tau_{m+1})+\frac{\partial u}{\partial x}(x_n,\tau_{m-1})\right]\\
& + \frac{1}{2}\left[\frac{\partial u}{\partial x}(x_n,\tau_{m+1})-u_{x_{n}}^{m+1}\right]+\frac{1}{2}\left[\frac{\partial u}{\partial x}(x_n,\tau_{m-1})-u_{x_{n}}^{m-1}\right]\\
& = O(\delta \tau^2+\delta x^4).
\end{split}
\end{equation*}
Hence, differential operator $\mathbb{D}$ in the PIDE~(\ref{eq:pidediscre}) can be approximated by discrete operator $\mathbb{D}_{\delta}$ with error at each mesh point $(x_n,\tau_m)$
\begin{equation}
\label{eq:differror}
\mathbb{D} u(x_n,\tau_m)-\mathbb{D}_{\delta}\left(\frac{u(x_n,\tau_{m+1})+u(x_n,\tau_{m-1})}{2}\right)=O(\delta \tau^2+\delta x^4).
\end{equation}
The integral operator of PIDE (\ref{eq:pidediscre}) is approximated by fourth-order accurate composite Simpson's rule (as discussed in Sec.~\ref{sec:disc_prob}). From Eq.~(\ref{eq:simp}), we have
\begin{equation}
\label{eq:interror}
\mathbb{I}u(x_n,\tau_m)-\mathbb{I}_{\delta}u(x_n,\tau_m)=O(\delta x^4).
\end{equation}
From Eqs.~(\ref{eq:timeerr}), ~(\ref{eq:differror}) and~(\ref{eq:interror}), result follows.
\end{proof}
\subsection{Stability}
\label{ssec:stability}
The stability of proposed compact finite difference method is proved using von Neumann stability analysis. Consider a single node
\begin{equation}
\label{eq:u_fou}
U_{n}^{m}=p^{m}e^{In\theta},
\end{equation}
where $I=\sqrt{-1}$, $p^{m}$ is the $m^{th}$ power of amplitude at time levels $\tau_{m}$. We consider the integration term given in Eq.~(\ref{eq:all_operator}) in an equivalent form as follows.
\[
\mathbb{I}u(x,\tau) = \lambda \int_{-L}^{L}u(y+x,\tau)g(y)dy.
\]
Fourth-order accurate composite Simpson's rule for above equation is then given by
\begin{equation*}
\begin{split}
\mathbb{I}_{\delta}u&=\delta x \sum_{k=0}^{N}w_{k}U_{k+n}^{m}g_{k},\\
 &= \delta x \sum_{k=0}^{N}w_{k}p^{m}e^{I\theta(k+n)}g_{k},\\
&\equiv p^{m}e^{I\theta n}G_{k},
\end{split}
\end{equation*}
where
\begin{equation}
\label{eq:g_k}
G_{k}=\delta x\sum_{k=0}^{N}w_{k}e^{I\theta k}g_{k} \:\:\: \mbox{and} \:\:\: g_{k}=g(x_{k}).
\end{equation}
The fourth-order accuracy of the numerical quadrature $G_{k}$ is proved in the following lemma.
\begin{lemma}
\label{lemma:1}
The numerical quadrature $G_{k}$ given in Eq.~(\ref{eq:g_k}) satisfies the following
\[
\lvert G_{k}\rvert \leq 1+ c \delta x^4,
\]
where $c$ is a constant.
\end{lemma}
\begin{proof}
Using the property of density function $g(x)$, we have
 \begin{equation}
\int_{\Omega}g(x)dx \leq \int_{-\infty}^{\infty}g(x)dx=1.
\end{equation}
Applying composite Simpson's rule to the above relation, we have
\[
\delta x \sum_{k=0}^{N}w_{k}g_{k} \leq 1+ c \delta x^4.
\]
From Eq~(\ref{eq:g_k}), we get the desired result as follows
\begin{equation}
\begin{split}
\lvert G_{k}\rvert &= \lvert \delta x \sum_{k=0}^{N}w_{k}e^{i \theta k}g_{k} \rvert,\\
&  \leq 1+ c \delta x^4.
\end{split}
\end{equation}
\end{proof}
For sake of simplicity, we denote $\frac{\sigma^2}{2}=a$ and $\left(r-\frac{\sigma^2}{2}-\lambda \zeta\right)=b$ in the rest of the section. Therefore, the fully discrete problem~(\ref{eq:corec3}) can be written as follows
\begin{equation}
\begin{split}
(1-2 a \delta \tau \Delta_{x}^2+\delta \tau(r+\lambda))U^{m+1}_{n}&=(1+2 a \delta \tau 2 \Delta_{x}^2-\delta \tau(r+\lambda))U^{m-1}_{n}+2 \delta \tau\left[\frac{b}{2}-\frac{a}{2}\Delta_{x}\right]U_{x_{n}}^{m+1}\\
&+2 \delta \tau\left[\frac{b}{2}-\frac{a}{2}\Delta_{x}\right]U_{x_{n}}^{m-1}+2 \delta\tau\lambda G_{k}U^{m}_{n}.
\label{eq:stab}
\end{split}
\end{equation}
The following relations are obtained from Eqs.~(\ref{eq:integer1}) and~(\ref{eq:integer2}) discussed in Sec.~\ref{sec:compact}.
\begin{equation}
\label{eq:stab1}
\Delta_{x}U_{n}^{m}=I \frac{sin(\theta)}{\delta x}U^{m}_{n},
\end{equation}
\begin{equation}
\label{eq:stab2}
\Delta_{x}^{2}U_{n}^{m}= \frac{2cos(\theta)-2}{\delta x^2}U^{m}_{n},
\end{equation}
\begin{equation}
\label{eq:stab3}
U_{x_{n}}^{m}= I \frac{3sin(\theta)}{\delta x(2+cos(\theta))}U^{m}_{n}.
\end{equation}
Using relation (\ref{eq:stab1}), (\ref{eq:stab2}) and (\ref{eq:stab3}) in the difference scheme (\ref{eq:stab}), we get
\begin{equation}
\begin{split}
\label{eq:stab4}
\left[1-4a\delta \tau \left(\frac{cos(\theta)-1}{\delta x^2}\right)+\delta \tau(r+\lambda)\right]U^{m+1}_{n}
&=\left[1+4a\delta \tau \left(\frac{cos(\theta)-1}{\delta x^2}\right)-\delta \tau(r+\lambda)\right]U^{m-1}_{n}\\
&+\delta \tau\left[\left(a\frac{sin(\theta)}{\delta x}+Ib\right)\frac{3sin(\theta)}{\delta x(2+cos(\theta))}\right]U_{n}^{m+1}\\
&+\delta \tau\left[\left(a\frac{sin(\theta)}{\delta x}+Ib\right)\frac{3sin(\theta)}{\delta x(2+cos(\theta))}\right]U_{n}^{m-1}\\
&+2\delta \tau \lambda G_{k} U_{n}^{m},
\end{split}
\end{equation}
which implies
\small
\begin{equation}
\begin{split}
\label{eq:stab5}
\left[1-\delta \tau a\frac{cos^{2}(\theta)+4cos(\theta)-5}{\delta x^2(2+cos(\theta))}+\delta \tau(r+\lambda)-I \delta \tau b\frac{3sin(\theta)}{\delta x(2+cos(\theta))}\right]U^{m+1}_{n}
&=\left[1+\delta \tau a\frac{cos^{2}(\theta)+4cos(\theta)-5}{\delta x^2(2+cos(\theta))}\right.\\
&\left.-\delta \tau(r+\lambda)+I\delta \tau b\frac{3sin(\theta)}{\delta x(2+cos(\theta))}\right]U^{m-1}_{n}\\
&+2\delta \tau \lambda G_{k} U_{n}^{m}.
\end{split}
\end{equation}
\normalsize
Now using Eq.~(\ref{eq:u_fou}) in above and divide the above equation by $p^{m-1}e^{In\theta}$, we get the amplification polynomial
\begin{equation}
\label{eq:amplif}
\Theta(\delta x, \delta \tau, \theta)=\gamma_{0}p^2-2\gamma_{1}p-\gamma_{2},
\end{equation}
where
\begin{equation}
\label{eq:gammas}
\begin{split}
\gamma_{0}
& = \left[1-\delta \tau
\left(a\frac{cos^{2}(\theta)+4cos(\theta)-5}{\delta x^2(2+cos(\theta))}-(r+\lambda)+Ib\frac{3sin(\theta)}{\delta x(2+cos(\theta))}\right)\right],\\
\gamma_{1}
& = \lambda \delta \tau G_{k},\\
\gamma_{2}
& = \left[1+\delta \tau
\left(a\frac{cos^{2}(\theta)+4cos(\theta)-5}{\delta x^2(2+cos(\theta))}-(r+\lambda)+Ib\frac{3sin(\theta)}{\delta x(2+cos(\theta))}\right)\right].\\
\end{split}
\end{equation}
The following lemma is used from \cite{Jcstrik04} in order to prove the stability of the proposed compact finite difference method.
\begin{lemma}
\label{lemma:2}
A finite difference scheme is stable if and only if all the roots, $p_{u}$, of the amplification polynomial $\Theta$ satisfies the following condition:
\begin{itemize}
\item There is a constant $C$ such that $|p_{u}| \leq 1+C\delta \tau$.
\item There are positive constants $a_{0}$ and $a_{1}$ such that if $a_{0} <|p_{u}|\leq 1+C\delta \tau$ then $|p_{u}|$ is simple root and for any other root $p_{v}$, following relation holds
    \[
    \left|p_{v}-p_{u}\right|\geq a_{1},
    \]
    as $\delta x$, $\delta \tau$$\rightarrow0$.
\end{itemize}
\end{lemma}
\begin{proof}
For proof, see \cite{Jcstrik04}.
\end{proof}
We prove the above Lemma.~\ref{lemma:2} for the proposed compact finite difference method as follows.
\begin{theorem}
The fully discrete problem~(\ref{eq:corec3}) is stable in the sense of Von-Neumann for $\delta \tau$ $\leq$ $1/(2\lambda)$.
\end{theorem}
\begin{proof}
First, some properties of the coefficients $\gamma_{0}, \gamma_{1}$ and $\gamma_{2}$ of amplification polynomial $\Theta$ are proved. Using Lemma.~\ref{lemma:1} in Eq.~(\ref{eq:gammas}) it is observed that
\[
\lvert \gamma_{1} \rvert < \delta \tau \lambda.
\]
Now, we can write
\[
\lvert \gamma_{0} \rvert=\lvert (1-A)-IB \rvert,
\]
where
\[
A=a\frac{cos^{2}(\theta)+4cos(\theta)-5}{\delta x^2(2+cos(\theta))}-(r+\lambda),\:\:\: and \:\:\:B=b\frac{3sin(\theta)}{\delta x(2+cos(\theta))}.
\]
This implies
\[
\lvert \gamma_{0} \rvert^{2}=1+A^2-2A+B^2.
\]
Since $\frac{cos^{2}(\theta)+4cos(\theta)-5}{\delta x^2(2+cos(\theta))} <0$, $a > 0$, $(r+\lambda)>0$ $\implies$ $A < 0$, therefore $\lvert \gamma_{0} \rvert > 1 $. Similarly
\[
\left \lvert \frac{\gamma_{2}}{\gamma_{0}} \right \rvert^{2}=\frac{1+A^2+2A+B^2}{1+A^2-2A+B^2}.
\]
 Again $A < 0$ $\implies$ $\left \lvert \frac{\gamma_{2}} {\gamma_{0}} \right \rvert < 1$. Now, roots of the amplification polynomial $\Theta$ can be written as
\begin{equation}
\begin{split}
\lvert p\rvert & = \left \lvert\frac{\gamma_{1}\pm \sqrt{\gamma_{1}^2-\gamma_{0}\gamma_{2}}}{\gamma_{0}}\right \rvert, \\
& \leq \left \lvert\frac{\gamma_{2}}{\gamma_{0}}\right \rvert^{\frac{1}{2}}+ 2\left \lvert\frac{\gamma_{1}}{\gamma_{0}}\right\rvert,\\
& \leq 1+2\delta \tau\lambda.
\end{split}
\end{equation}
Hence, first part of the Lemma~\ref{lemma:2} is proved for constant $C=2\lambda$. Now for second part of the Lemma.~\ref{lemma:2}, let us assume that $p_{1}$ and $p_{2}$ are two roots of amplification polynomial $\Theta$. Take the constant $a_{0}=1$ which will imply that $p_{1}>1$, then
\begin{equation}
\begin{split}
|p_{1}-p_{2}| & \geq 2|p_{1}|-|p_{1}+p_{2}|,\\
& \geq 2-2\delta \tau\lambda.
\end{split}
\end{equation}
If $\delta \tau$ satisfies the given condition, we have
\[
|p_{1}-p_{2}|\geq 1,
\]
and this prove the second part of the Lemma.~\ref{lemma:2} with $a_{1}=1$. This completes the proof.
\end{proof}
\section{Numerical Results}
\label{sec:numerical}
\par In this section, the applicability of the proposed compact finite difference method for pricing European and American options under jump-diffusion models is demonstrated. According to \cite{thomee70}, fourth-order convergence cannot be expected for non-smooth initial conditions. Since the initial conditions given in Equations~(\ref{eq:initial_eur}) and~(\ref{eq:initial_amer}) have low regularity, the smoothing operator $\phi_4$ given in \cite{thomee70} is employed to smoothen the initial conditions and it's Fourier transform is define as
\[
\hat{\phi}_{4}(\omega)=\left(\frac{sin(\omega/2)}{\omega/2}\right)^4\left[1+\frac{2}{3}sin^2(\omega/2)\right].
\]
As a result, the following smoothed initial condition $(\tilde{u}_{0})$ is obtained
\begin{equation}
\label{eq:smoothed}
\tilde{u}_{0}(x_1)=\frac{1}{\delta x}\int_{-3\delta x}^{3\delta x}\phi_{4}\left(\frac{x}{\delta x}\right)u_{0}(x_1-x)dx,
\end{equation}
where $u_{0}$ is the actual non-smooth initial condition and $x_{1}$ is the grid point where smoothing is required. The smoothed initial conditions obtained from Equation~(\ref{eq:smoothed}) tends to the original initial conditions as $\delta x\rightarrow 0$. The parameters considered for pricing European and American options under Merton jump-diffusion model are listed in Table~\ref{table:parameter}. The parabolic mesh ratio $(\frac{\delta \tau}{\delta x^2})$ is fixed as $0.4$ in all our computations, although neither the von Neumann stability analysis nor the numerical experiments showed any such restriction. The relative $\ell^2$-error $\frac{||U_{ref}-U||_{\ell^2}}{||U_{ref}||_{\ell^2}}$ is used to determine the numerical convergence rate, where $U_{ref}$ represents the numerical solution on a fine grid $(\delta x = 4.8828125e-04)$ and $U$ denotes the numerical solution on coarser grid. Order of convergence is obtained as the slope of the linear least square fit of the individual error points in the loglog plot of error versus number of grid points.
\begin{table}[h!]
	\begin{center}
		\begin{tabular}{ P{2 cm} P{5 cm} }
			\hline
			& \hspace{-3cm}European and American options \\
			\hline
			Parameters & Values \\
			\hline
			$\lambda$ & $0.10$  \\
			$T$ & $0.25$ \\
			$r$ & $0.05$  \\
			$K$ & $100$ \\
			$\sigma$ & $0.15$ \\
			$\mu_{J}$ & $-0.90$ \\
			$\sigma_{J}$ & $0.45$  \\
			$S_0$ & $100$ \\
			\hline
		\end{tabular}
	\end{center}
	\caption{The values of parameters for pricing European and American options under jump-diffusion models.}
	\label{table:parameter}
\end{table}
\par In option pricing, Greeks are important instruments for the measurement of an option position's risks. The rate of change of option price with respect to change in the underlying asset's price is known as Delta whereas the rate of change in the delta with respect to change in the underlying price is called as Gamma. The proposed compact finite difference method is considered for valuation of options and Greeks as well in the following examples.
%%%%%%%%%%%%%%%% First Example %%%%%%%%%%%%%%%%%%%%%%%
\begin{example}(Merton jump-diffusion model for European put options with constant volatility)
\end{example}
\begin{figure}[h!]
	\begin{center}
		\subfigure[]{%
			\includegraphics[scale=0.50]{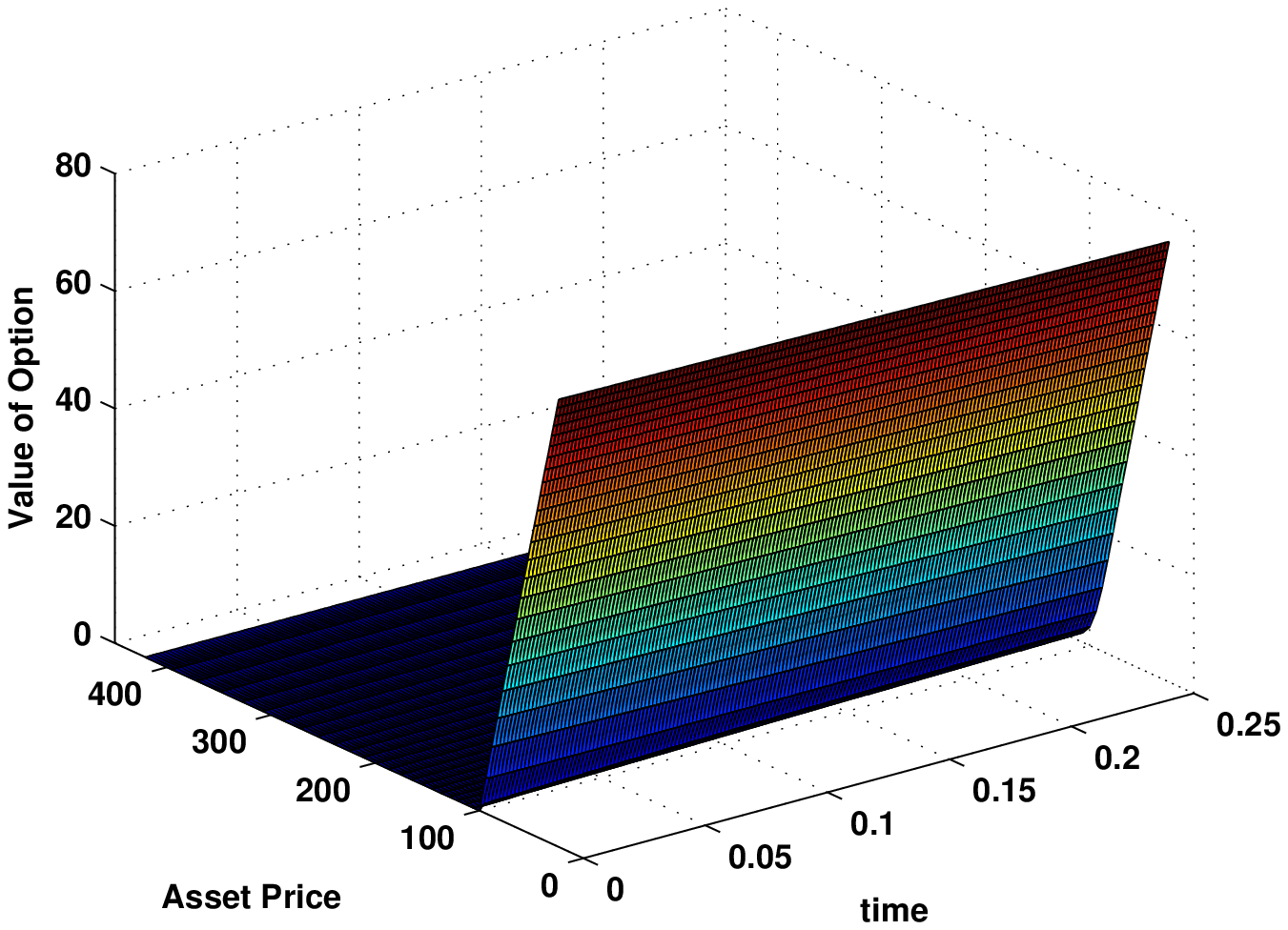}
			\label{fig:price_euro_three1}}%
	     \subfigure[]{%
		\includegraphics[scale=0.50]{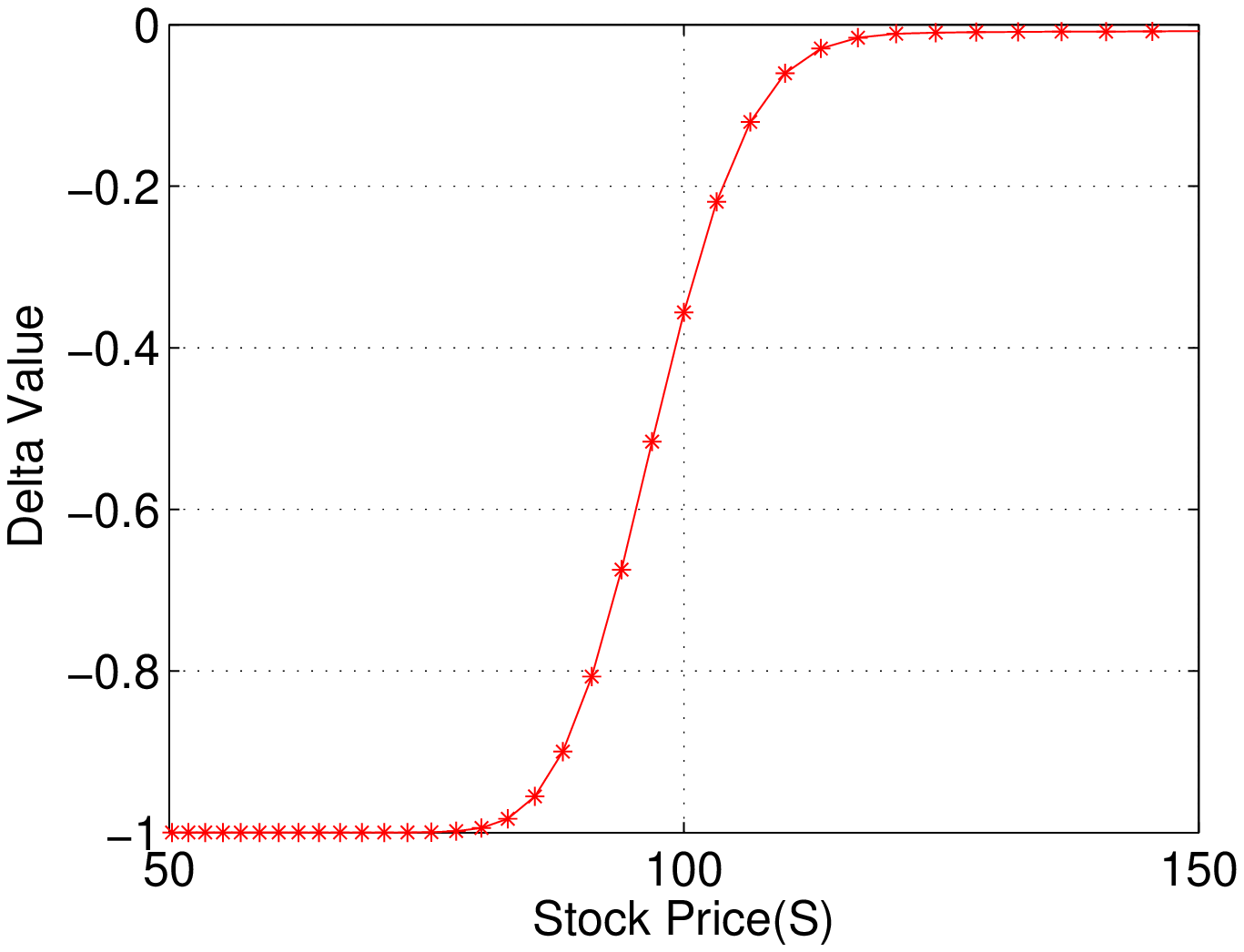}
		\label{fig:delta}}
	     \subfigure[]{%
		\includegraphics[scale=0.50]{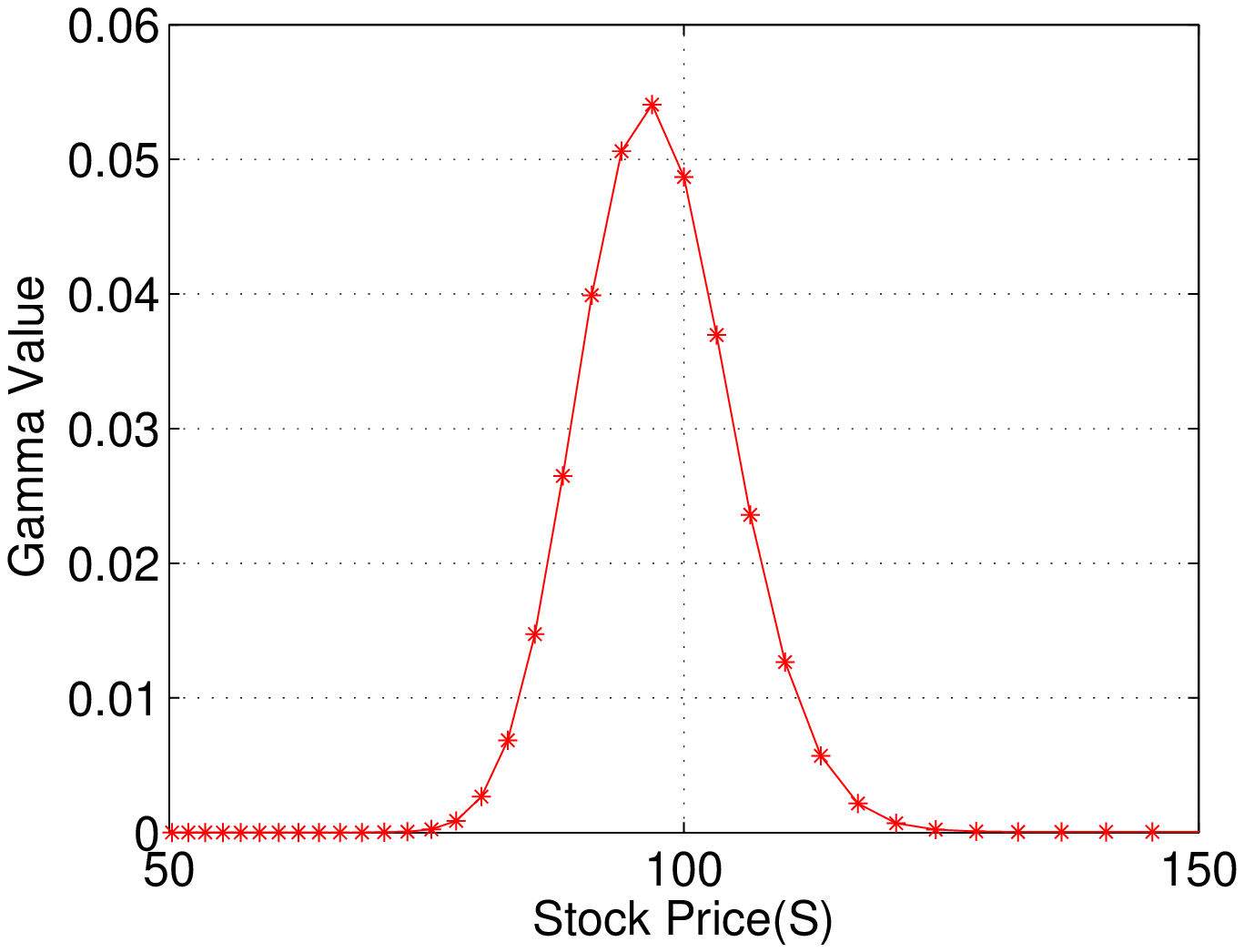}
		\label{fig:gamma}}%
			 \subfigure[]{%
		\includegraphics[scale=0.38]{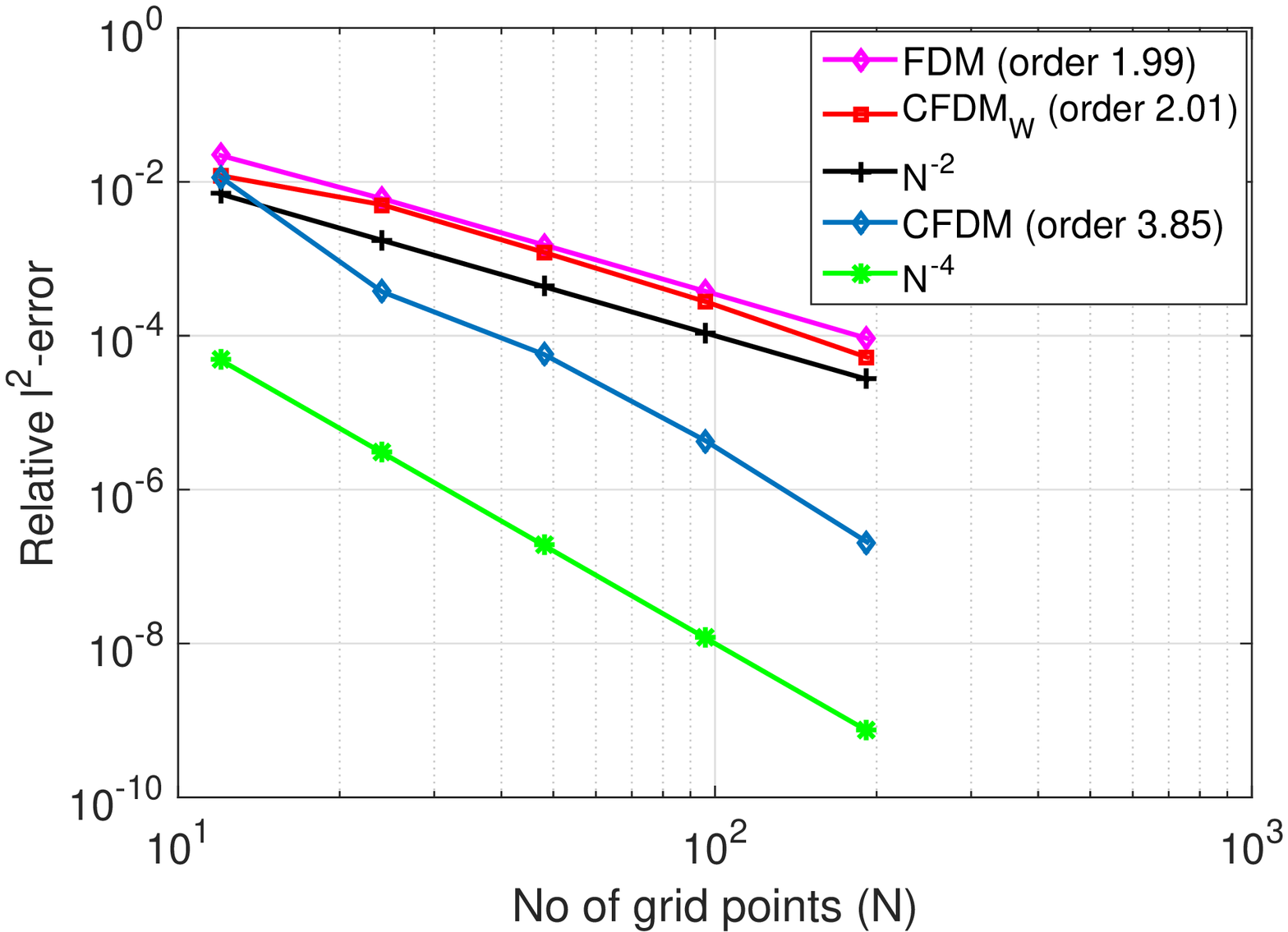}
		\label{fig:rate_european1}}	
		\caption{(a) Prices of European put options as function of stock price and time, (b) Delta of European put options versus stock price, (c) Gamma of European put options versus stock price, and (d) Relative $\ell^2$ error using (i) FDM: finite difference method, (ii). $CFDM_W$: proposed compact finite difference method with non-smooth initial condition, (iii). CFDM: proposed compact finite difference method with smooth initial condition.}
		\label{fig:european1}
	\end{center}
\end{figure}
\begin{table}[h!]
	\begin{tabular}{ m{1cm} | m{4cm} | m{4cm}| m{4cm} }
		\hline
	(S, $\tau$)	& \hspace{1cm}Option Price & \hspace{1cm}Delta & \hspace{1cm}Gamma \\
		\hline
   &  In \cite{Kad16} \hspace{1cm} Our method  &  In \cite{Kad16} \hspace{1cm} Our method &  In \cite{Kad16} \hspace{1cm} Our method\\ 	
		\hline
		 (90,0) & 9.285418 \hspace{1cm} 9.285416 & -0.846715 \hspace{0.8cm} -0.846716 & 0.034860 \hspace{1cm} 0.034862 \\
		\hline
		(100,0) & 3.149026 \hspace{1cm} 3.149018 & -0.355663 \hspace{0.8cm} -0.355661 & 0.048825 \hspace{1cm} 0.048828 \\
		\hline
		(110,0) & 1.401186 \hspace{1cm} 1.401182 & -0.058101 \hspace{0.8cm} -0.058103 & 0.012129 \hspace{1cm} 0.012131 \\
		\hline
	\end{tabular}
	\caption{Values of European put options and Greeks under Merton jump-diffusion model with constant volatility using $N=1536$.}
	\label{table:european1}
\end{table}
\begin{figure}[h!]
	\begin{center}
		\subfigure[]{%
			\includegraphics[scale=0.50]{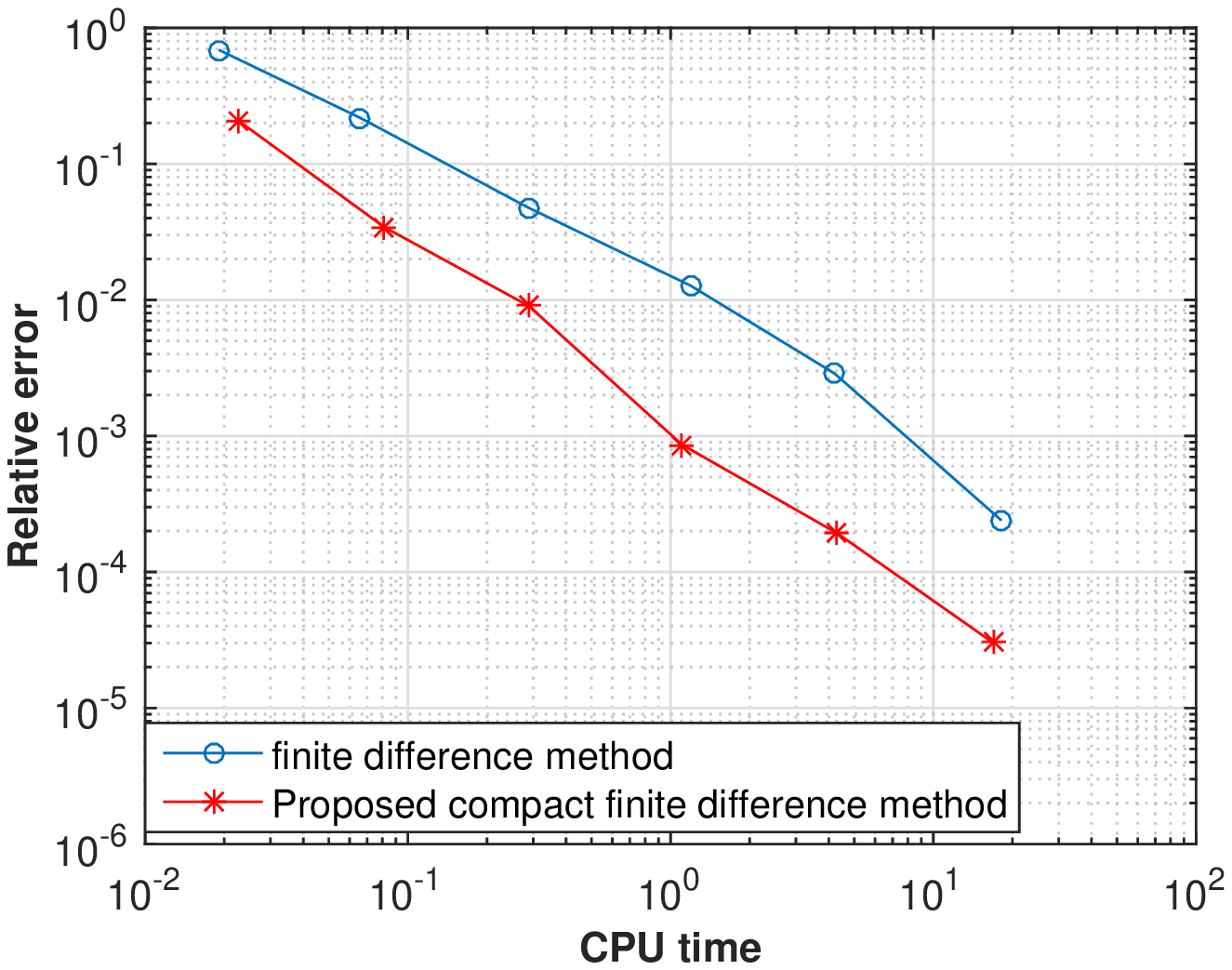}
			\label{fig:cpu_time}}%
		\subfigure[]{%
			\includegraphics[scale=0.50]{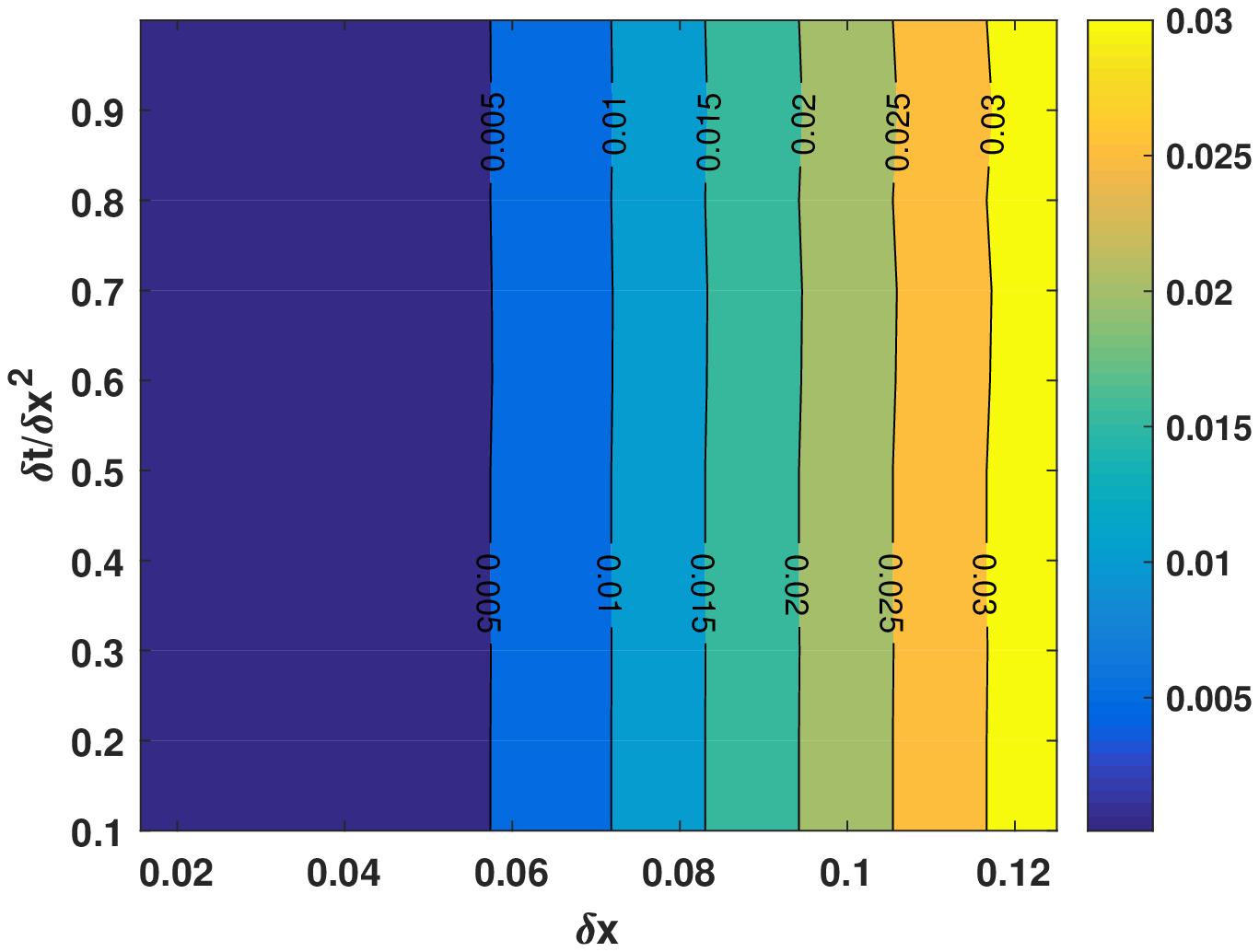}
			\label{fig:numer_stability}}
		\caption{(a) Efficiency: CPU time and relative error for  finite difference method and proposed compact finite differential method (b) Numerical Stability Plot.}
		\label{fig:cpu_stab}
	\end{center}
\end{figure}
 \begin{figure}[h!]
	\begin{center}
		\subfigure[]{%
			\includegraphics[scale=0.500]{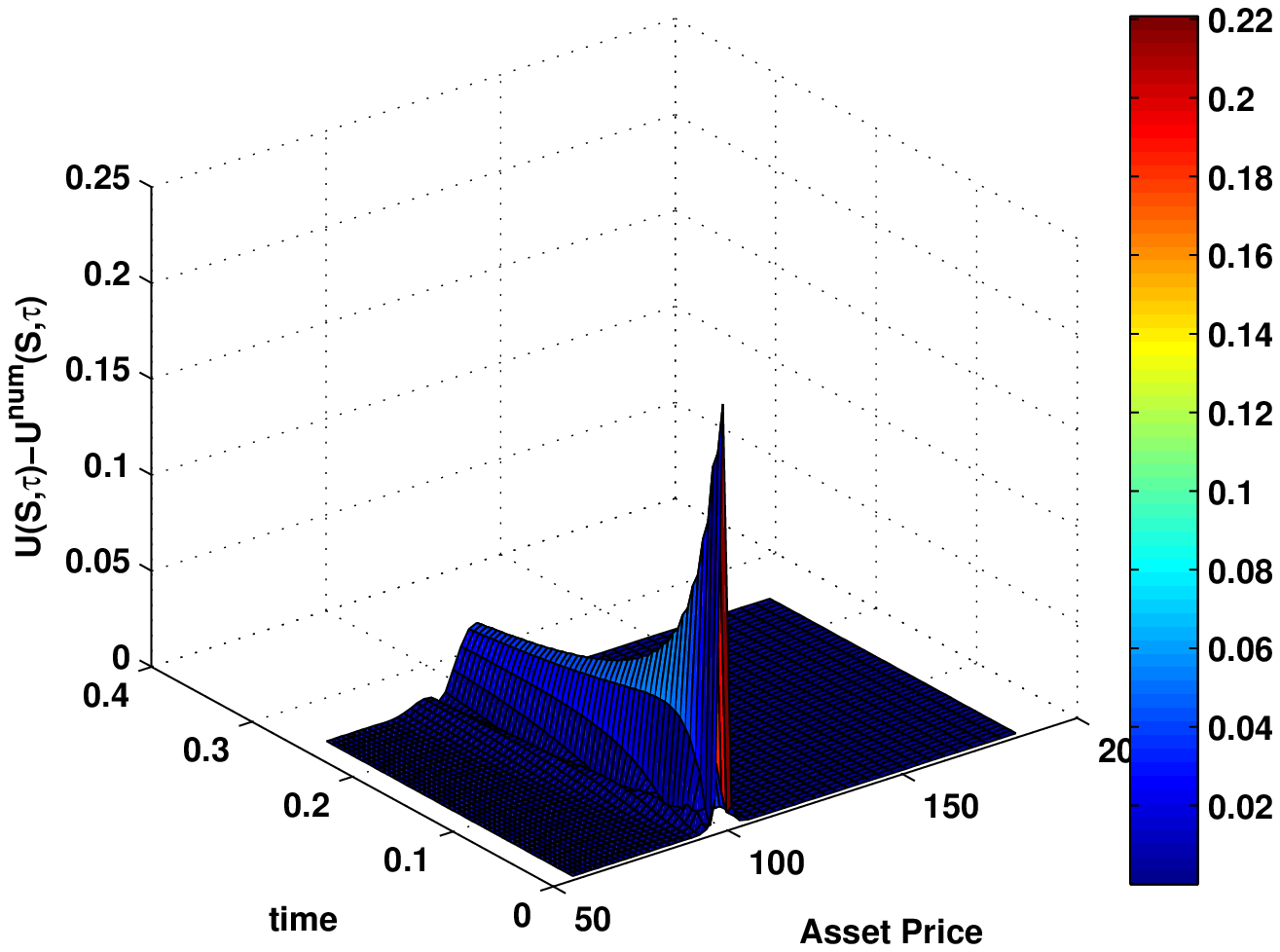}
			\label{fig:err_woutsmo_finite}}%
		\subfigure[]{%
			\includegraphics[scale=0.500]{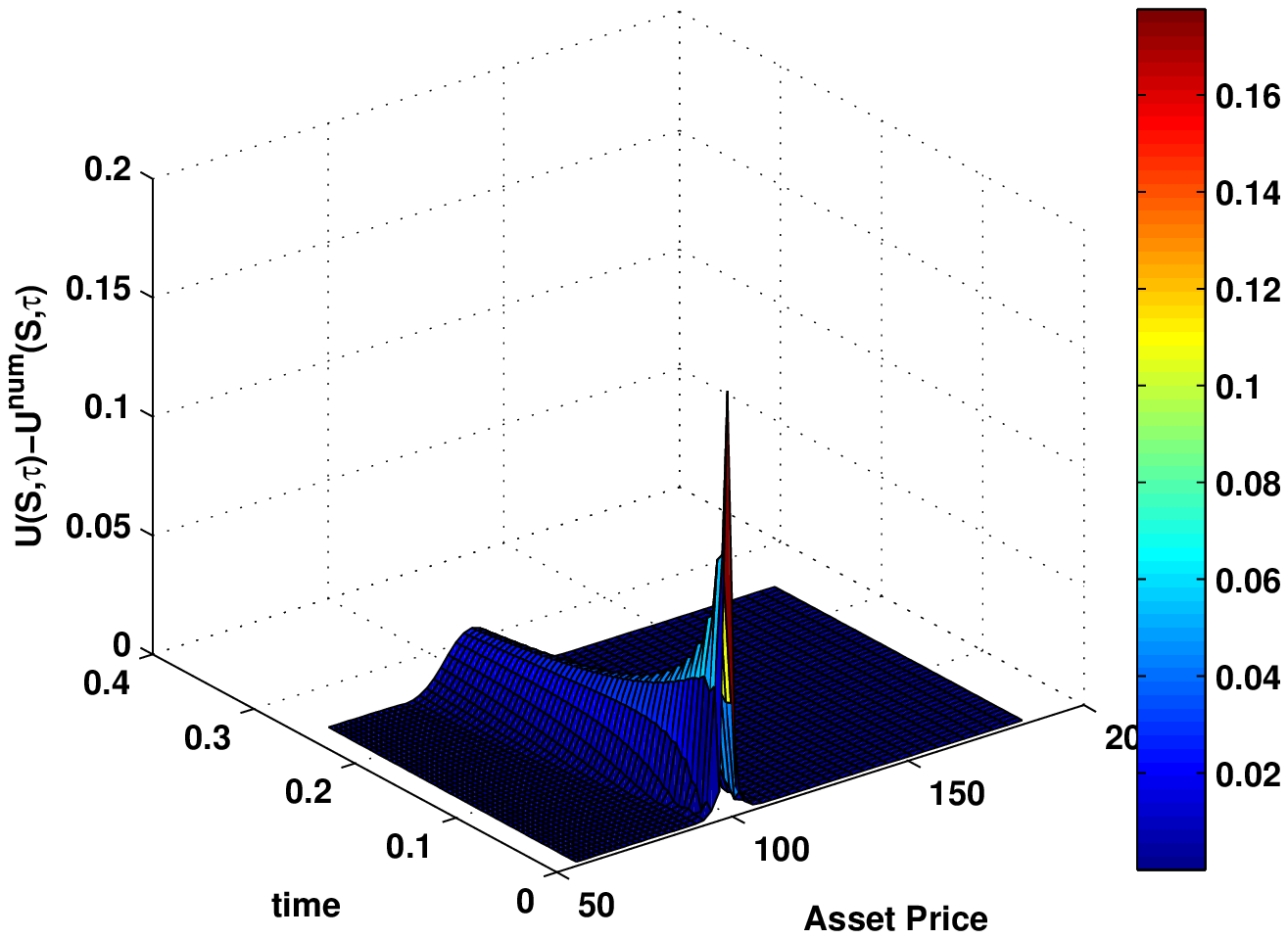}
			\label{fig:err_woutsmo_compact}}
		\subfigure[]{%
			\includegraphics[scale=0.500]{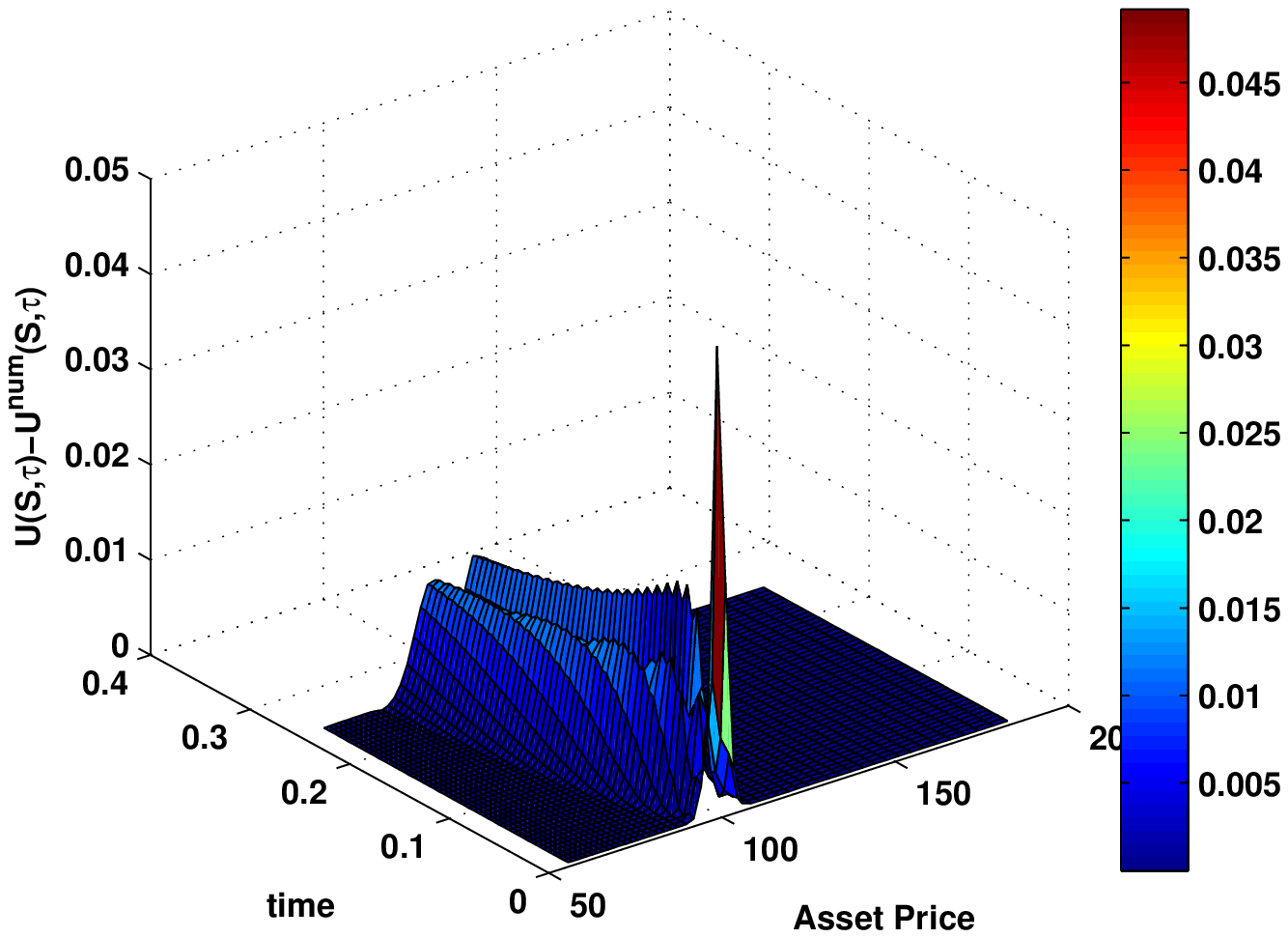}
			\label{fig:err_smo_finite}}%
		\subfigure[]{%
			\includegraphics[scale=0.500]{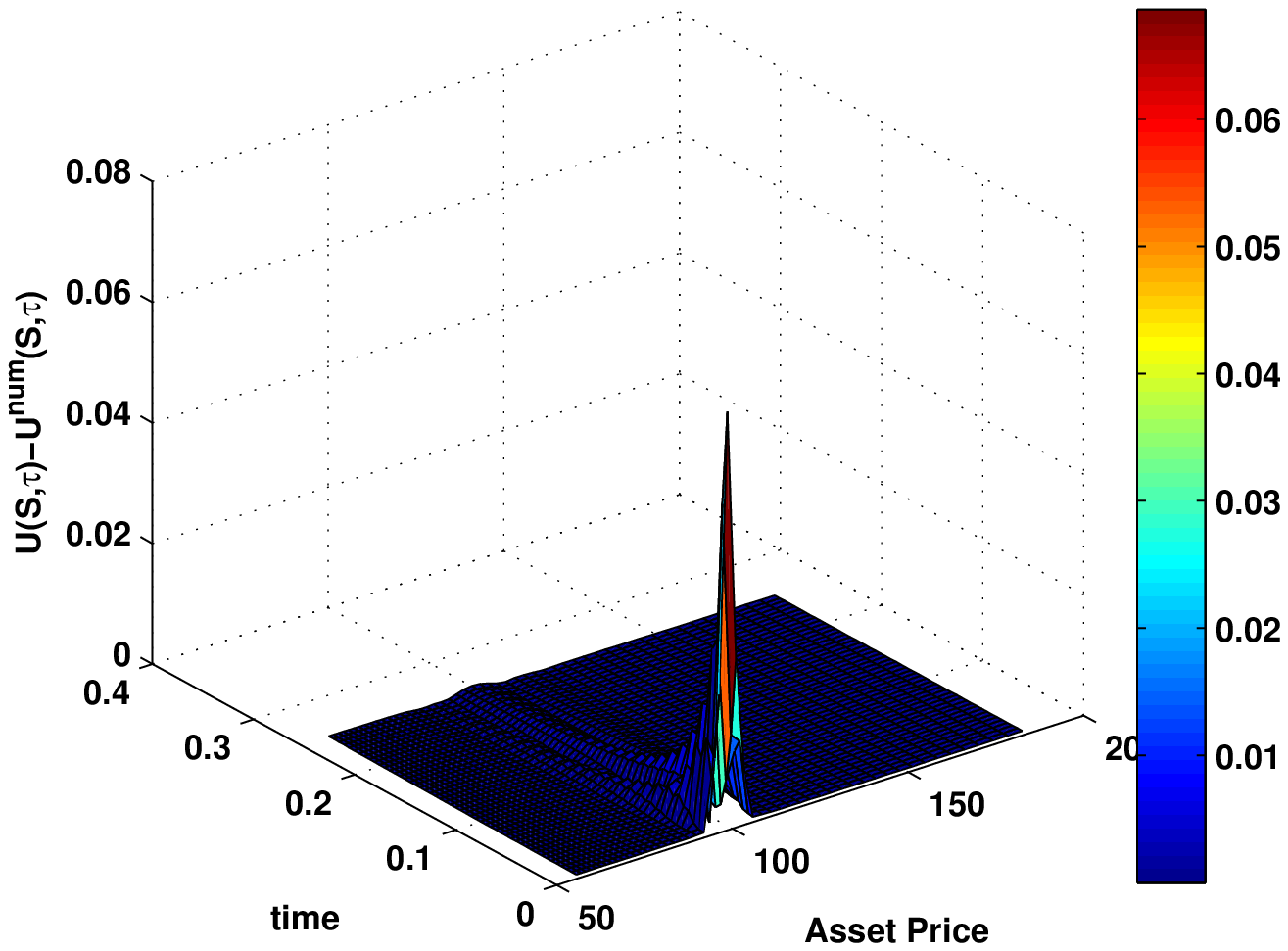}
			\label{fig:err_smo_compact}}
		\caption{The difference between reference and numerical solutions as a function of asset price and time using: (a) finite difference method with non-smooth initial condition, (b) proposed compact finite difference method with non-smoothing initial condition, (c) finite difference method with smoothed initial condition and (d) proposed compact finite difference method with smoothed initial condition.}
		\label{fig:err_smo}
	\end{center}
\end{figure}
\par The values of option prices and Greeks for various stock prices are presented in Table~\ref{table:european1} and it is observed that proposed compact finite difference method is accurate for valuation of options and Greeks as well. Prices of European options and Greeks are plotted in Figs.~\ref{fig:price_euro_three1},~\ref{fig:delta} and~\ref{fig:gamma} respectively. The relative $\ell^2$-errors using finite difference method (second-order accurate) and proposed compact finite difference method are plotted in Fig.~\ref{fig:rate_european1} and it can be concluded that proposed method is only second order accurate with non-smooth initial condition. Further, it is observed that numerical order of convergence rate is in excellent agreement with the theoretical order of convergence of the proposed method when initial condition is smoothed.
\par The PIDE~(\ref{eq:pidediscre}) is also solved using finite difference method \cite{Kwon11} in order to compare the efficiency of proposed compact finite difference method with finite difference method. The relative $\ell^2$ errors between the numerical and reference solutions and corresponding CPU time at grid points $N$=$12, 24, 48, 96, 192$ and $384$ using finite difference method and proposed compact finite difference method are computed and presented in Fig.~\ref{fig:cpu_time}. It is observed from Fig.~\ref{fig:cpu_time} that for a given accuracy, proposed method is significantly efficient as compared to finite difference method. An additional numerical stability test is performed in order to validate the theoretical stability results. The numerical solutions for varying values of the parabolic mesh ratio $(\frac{\delta \tau}{\delta x^2})$ and mesh width $\delta x$ are computed.
Plotting the associated relative $\ell^{2}$ errors should allow us to detect stability restrictions depending on the values of $\delta \tau$ and $\delta x$. The similar approach for numerical stability test is also discussed in \cite{DurF12}. The relative $\ell^{2}$ error is plotted in Fig~\ref{fig:numer_stability} with $\frac{\delta \tau}{\delta x^2} = \frac{k}{10}$, $k = 1, . . . , 10$ for various values of $\delta x$ and it is observed that the influence of the parabolic mesh ratio on the relative $\ell^{2}$ error is only marginal. Thus, we can infer that there does not seem to be any condition on the choices of $\delta \tau$ and $\delta x$.
\par The difference between the reference and numerical solutions as a function of asset price and time with non-smooth initial condition are plotted in Figs.~\ref{fig:err_woutsmo_finite} and~\ref{fig:err_woutsmo_compact} respectively. It is observed from the figures that maximum error at strike price is comparatively smaller with proposed compact finite difference method. Similarly, the difference between reference and numerical solutions with smoothed initial conditions are plotted in Figs.~\ref{fig:err_smo_finite} and~\ref{fig:err_smo_compact}. It is evident from the figures that oscillations in the solution near the strike price are lesser with proposed compact finite difference method.
%%%%%%%%%%% SecondExample %%%%%%%%%%%%%%%%%%%%%%%%%%%%%%%%%%%%%%
\begin{example}(Merton jump-diffusion model for European put options with local volatility)
\end{example}
\par In this example, the volatility $\sigma$ is assumed to be a function of stock price and time and is given as
\begin{equation}
\label{eq:vola}
\sigma(x,\tau)=0.15+0.15\left(0.5+2(T-\tau)\right)\frac{\left((S_{0}e^{x}/100)-1.2\right)^2}{\left(S_{0}e^{x}/100\right)^2+1.44}.
\end{equation}
In Table~\ref{table:european}, the values of European options with local volatility for various stock prices are presented. It is observed that option prices obtained using proposed compact finite difference method are in excellent agreement with the reference values.  The values of European options as a function of stock price and time are plotted in Figs.~\ref{fig:price_euro_three}. The relative $\ell^2$-errors using finite difference method and proposed compact finite difference method are plotted in Fig.~\ref{fig:rate_european} and it is observed that proposed method is only second order accurate with non-smooth initial condition. The numerical order of convergence rate agrees with the theoretical order of convergence rate of the proposed method when the initial condition is smoothed.
\begin{figure}[h!]
\begin{center}
\subfigure[]{%
\includegraphics[scale=0.50]{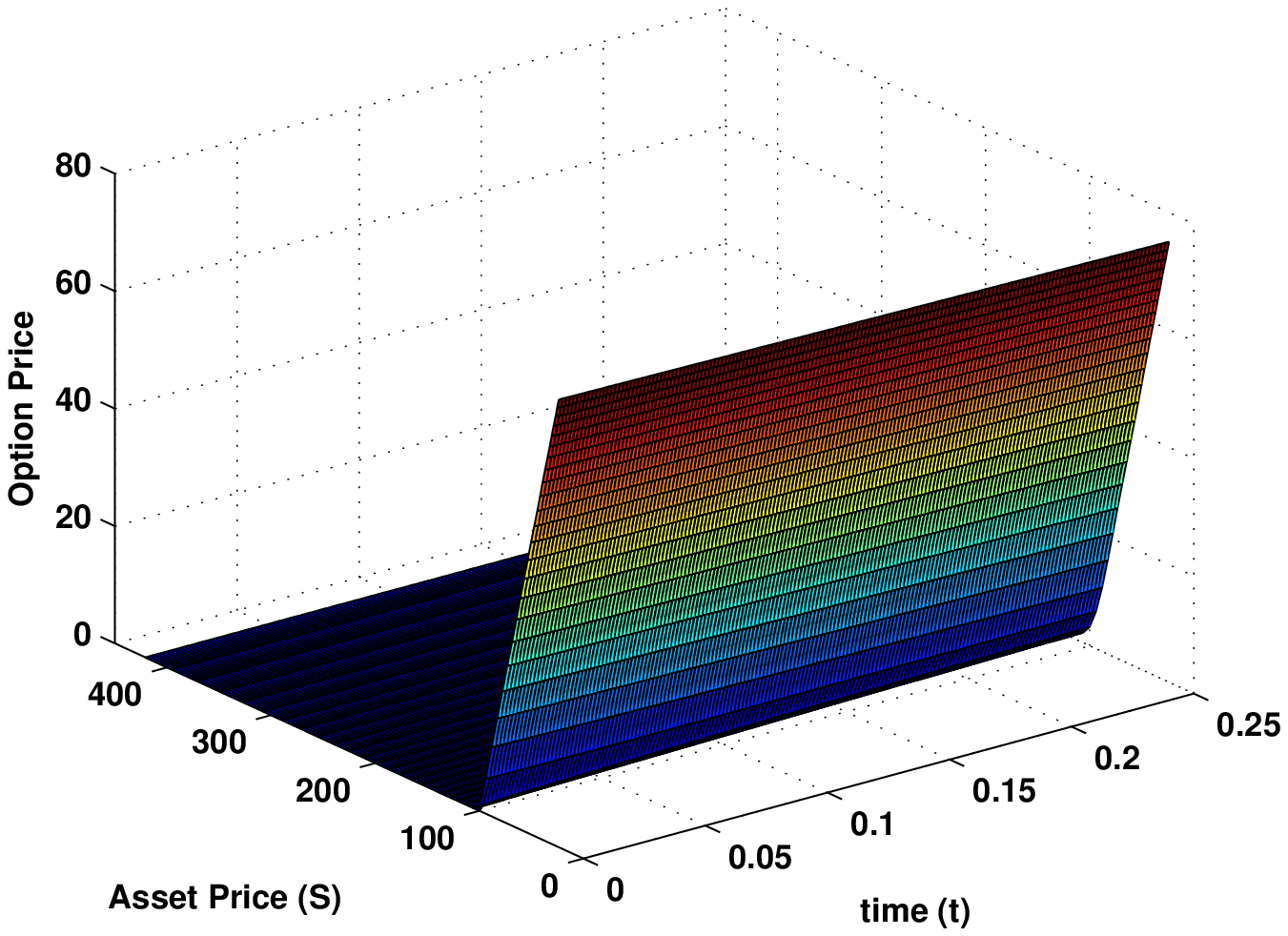}
\label{fig:price_euro_three}}%
\subfigure[]{%
	\includegraphics[scale=0.36]{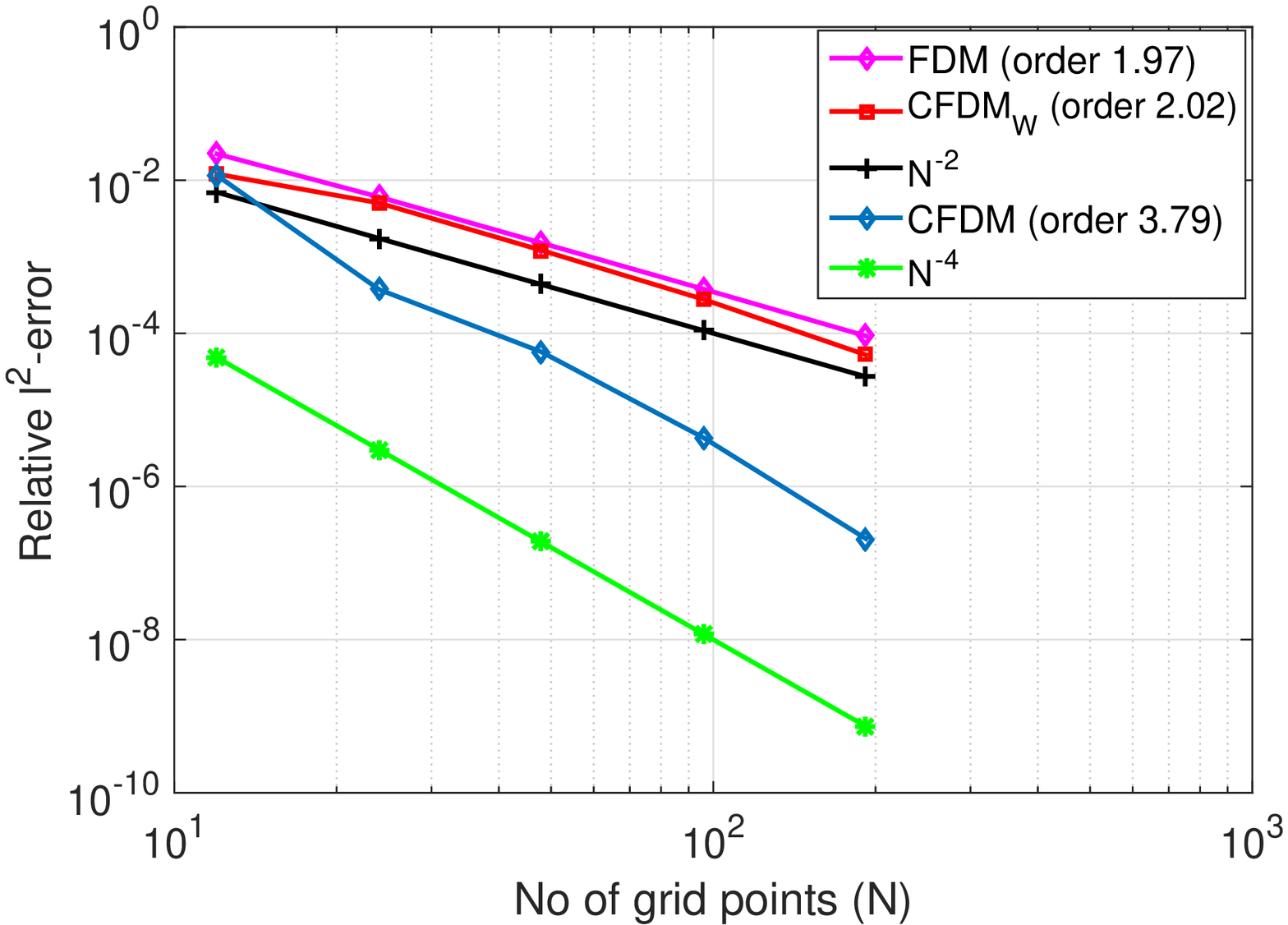}
	\label{fig:rate_european}}
\caption{(a) Prices of European put options with local volatility as a function of asset price and time, (b) Relative $\ell^2$ error with (i) FDM: finite difference method, (ii). $CFDM_W$: proposed compact finite difference method without smoothing the initial condition, (iii). CFDM: proposed compact finite difference method with smooth initial condition.}
\label{fig:european}
\end{center}
\end{figure}
\begin{table}[h!]
\begin{tabular}{ m{4cm} | m{3cm} | m{3cm}| m{3cm} }
\hline
  & (S, $\tau$) = (90,0) & (S, $\tau$) = (100,0) & (S, $\tau$) = (110,0) \\
 \hline
Reference values \cite{JLee15} & 9.317323 & 3.183681 & 1.407745 \\
 \hline
Proposed method & 9.317322 & 3.183682 & 1.407743 \\
 \hline
\end{tabular}
\caption{Values of European put options with local volatility under Merton jump-diffusion model using $N=1536$.}
\label{table:european}
\end{table}
%%%%%%%% Third Example %%%%%%%%%%%%%%%%%%%%%%%%%%%
\begin{example}(Merton jump-diffusion model for American put options with constant volatility)
\end{example}
\par The values of American options for various stock prices are presented in Table~\ref{table:american1} and it is observed that proposed compact finite difference method is also accurate for valuation of American options. The values of American options as a function of stock price and time are plotted in Figs.~\ref{fig:price_amer_three1}. The relative $\ell^2$-errors using finite difference method and proposed compact finite difference method are plotted in Fig.~\ref{fig:rate_amer1} and it is observed that proposed method is only second order accurate with non-smooth initial condition. The numerical order of convergence rate is $3.2$ with smoothed initial condition which does not represents the theoretical order of convergence rate. The reason could be the lack of regularity of the problem due to the free boundary feature which needs further research to be resolved \cite{Bastani13}.
\begin{figure}[h!]
\begin{center}
\subfigure[]{%
\includegraphics[scale=0.50]{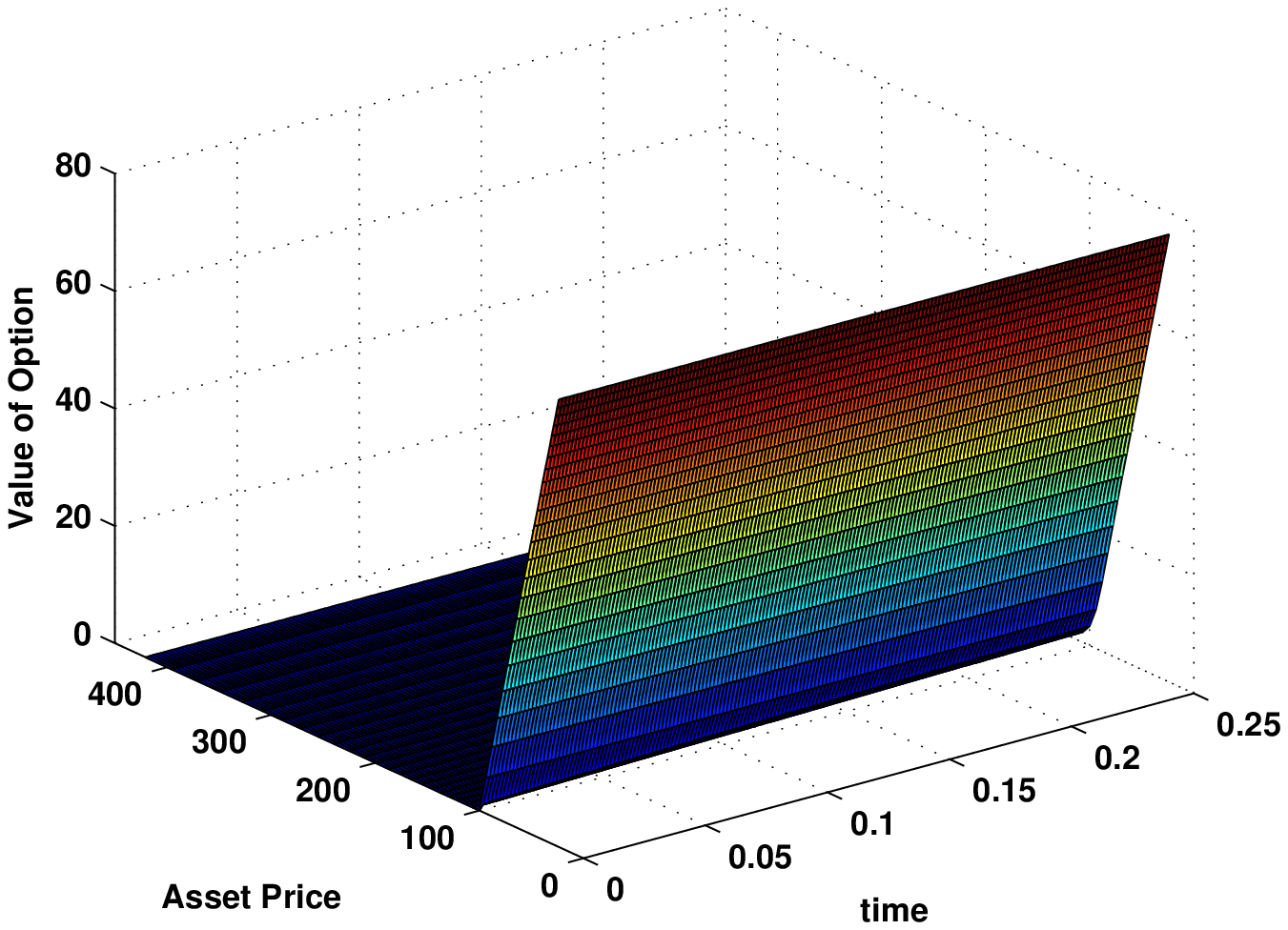}
\label{fig:price_amer_three1}}%
\subfigure[]{%
	\includegraphics[scale=0.36]{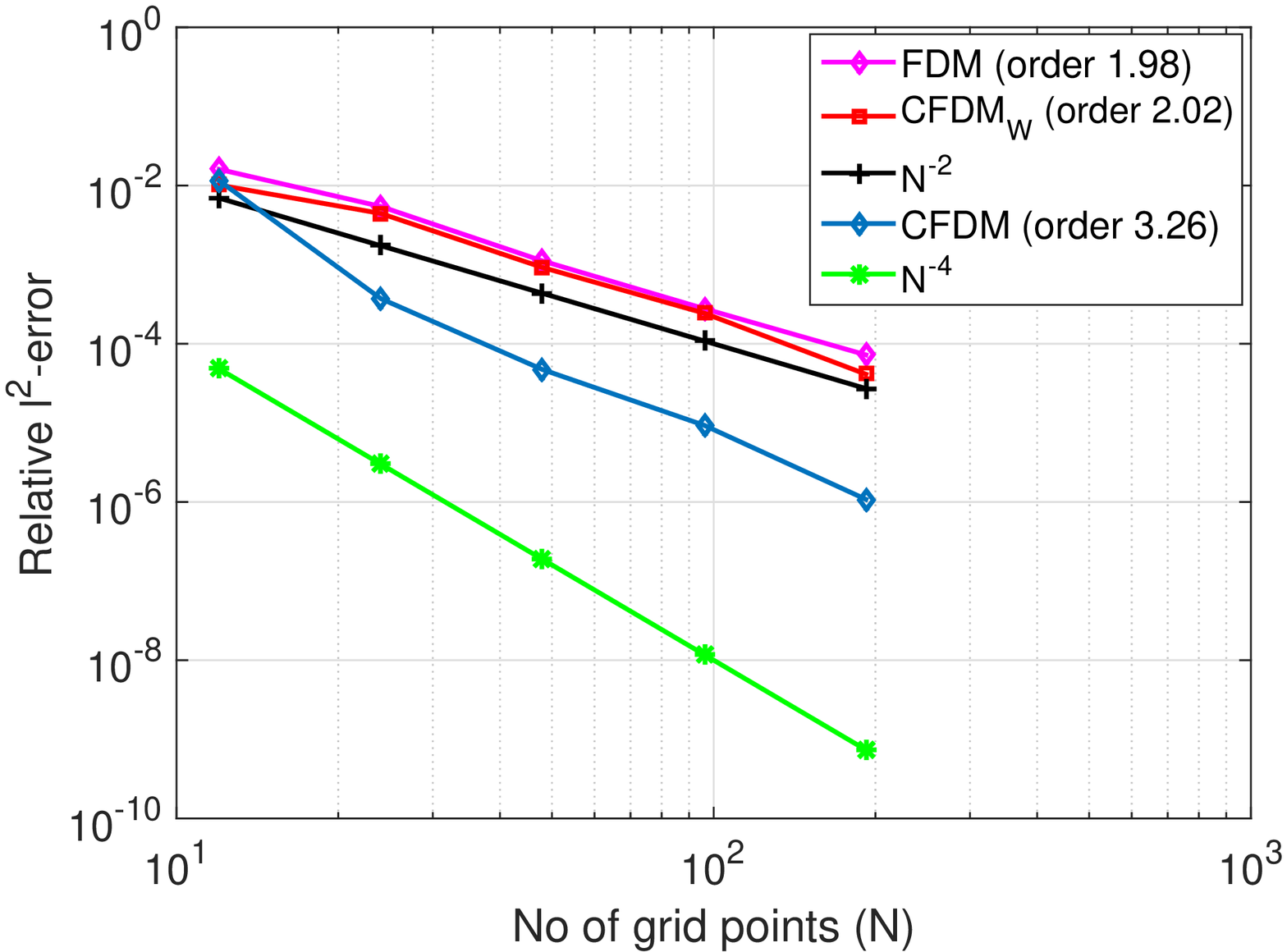}
	\label{fig:rate_amer1}}
\caption{(a) Prices of American put options as a function of asset price and time, (b) Relative $\ell^2$ error with (i) FDM: finite difference method, (ii). $CFDM_W$: proposed compact finite difference method without smoothing the initial condition, (iii). CFDM: proposed compact finite difference method with smooth initial condition.}
\label{fig:american1}
\end{center}
\end{figure}
\begin{table}[h!]
\begin{tabular}{ m{4cm} | m{3cm} | m{3cm}| m{3cm} }
\hline
  & (S, $\tau$) = (90,0) & (S, $\tau$) = (100,0) & (S, $\tau$) = (110,0) \\
 \hline
Reference values \cite{JLee15} & 10.003866 & 3.241207 & 1.419790 \\
 \hline
Proposed compact scheme & 10.003862 & 3.241208 & 1.419791\\
 \hline
\end{tabular}
\caption{Values of American put options under Merton jump-diffusion model with $N=1536$.}
\label{table:american1}
\end{table}
%%%%%%%% Fourth Example %%%%%%%%%%%%%%%%%%%%%%%%%%%%%%
\begin{example}(Merton jump-diffusion model for American put options with local volatility)
\end{example}
\begin{figure}[h!]
\begin{center}
\subfigure[]{%
\includegraphics[scale=0.50]{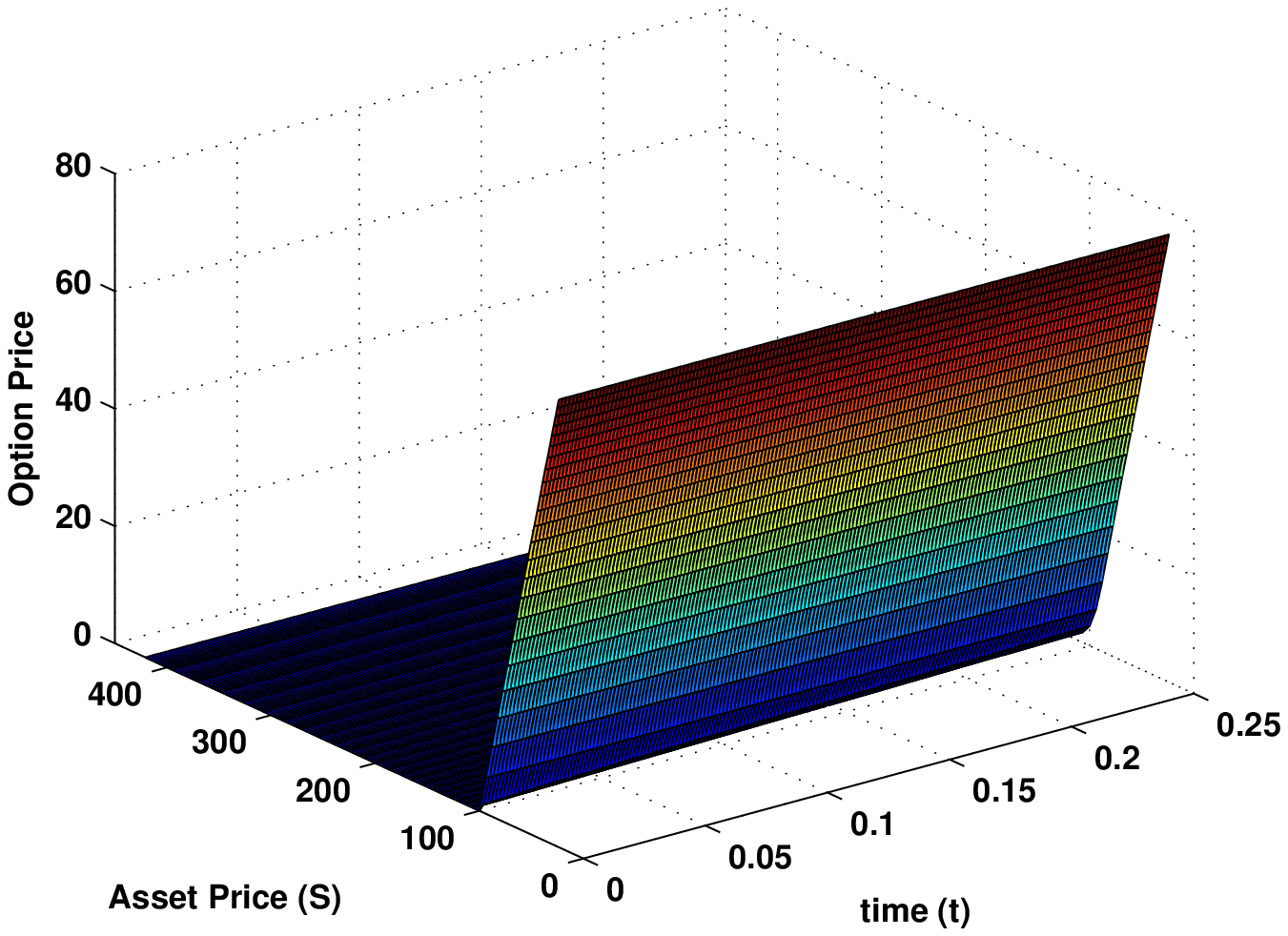}
\label{fig:price_amer_three}}%
\subfigure[]{%
	\includegraphics[scale=0.36]{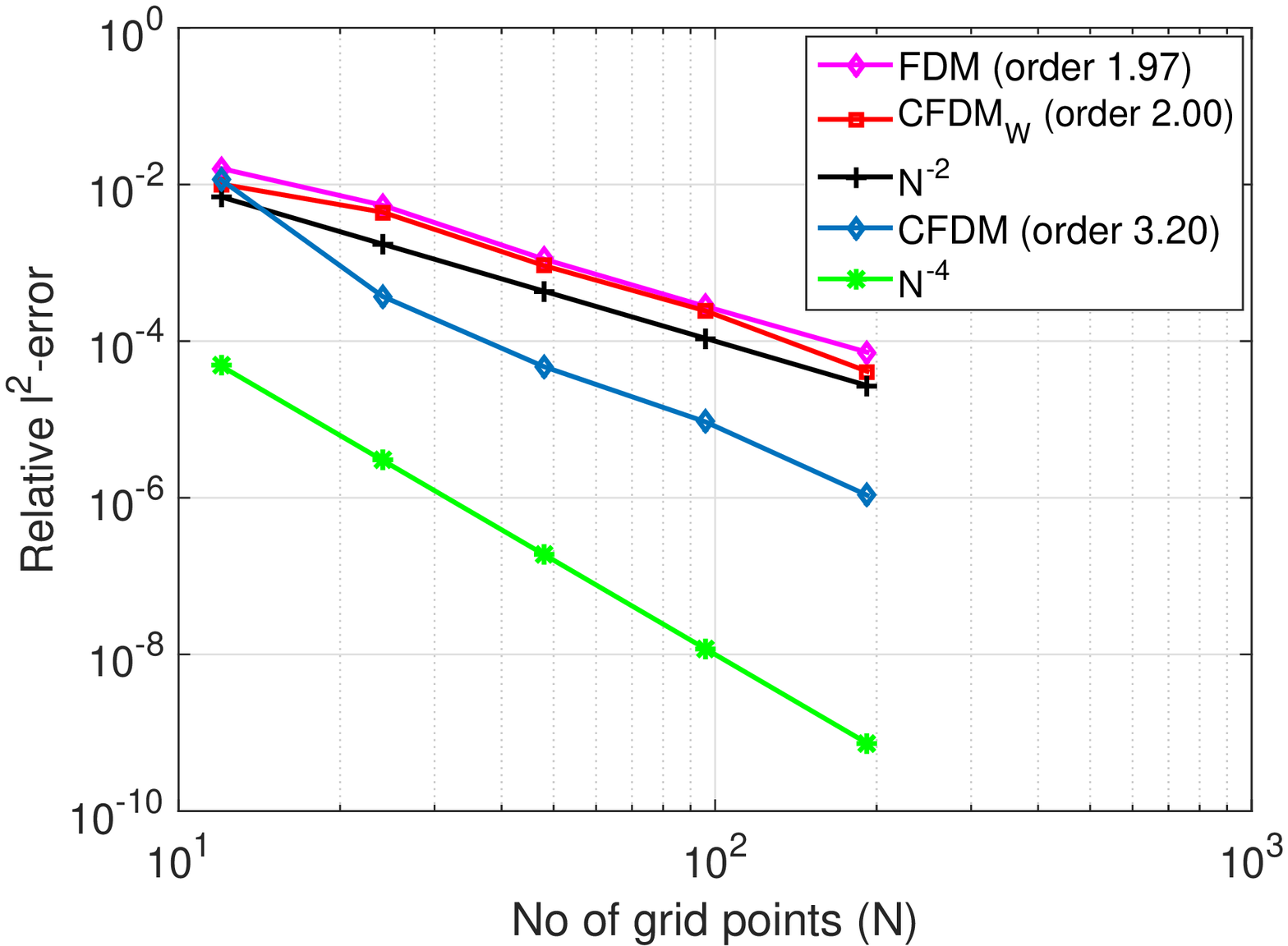}
	\label{fig:rate_amer}}
\caption{(a) Prices of American put options with local volatility as a function of asset price and time, (b) Relative $\ell^2$ error with (i) FDM: finite difference method, (ii). $CFDM_W$: proposed compact finite difference method without smoothing the initial condition, (iii). CFDM: proposed compact finite difference method with smooth initial condition.}
\label{fig:american}
\end{center}
\end{figure}
\begin{table}[h!]
	\begin{tabular}{ m{4cm} | m{3cm} | m{3cm}| m{3cm} }
		\hline
		& (S, $\tau$) = (90,0) & (S, $\tau$) = (100,0) & (S, $\tau$) = (110,0) \\
		\hline
		Reference values \cite{JLee15} & 10.008881 & 3.275957 & 1.426403 \\
		\hline
		Proposed compact scheme & 10.008880 & 3.275955 & 1.426403 \\
		\hline
	\end{tabular}
	\caption{Values of American put options with local volatility under Merton jump-diffusion model using $N=1536$.}
	\label{table:american}
\end{table}
\par In Table~\ref{table:american}, the values of American options with non-constant volatility are presented for various stock prices. It can be concluded that proposed method is also accurate for valuation of American options with non-constant volatility. Figs.~\ref{fig:price_amer_three} presents the values of American options as a function of stock price and time. The relative $\ell^2$-errors using finite difference method and proposed compact finite difference method are plotted in Fig.~\ref{fig:rate_amer}.
\section{Conclusion and future work}
\label{sec:conclu}
\par In this article, a compact finite difference method has been proposed for pricing European and American options under Merton jump-diffusion model with constant and local volatilities. Wave numbers and modified wave numbers for various difference approximations have been discussed and it is observed that compact approximations have better resolution characteristics as compared to finite difference approximations. Consistency and stability of fully discrete problem have also been proved. The effect of non-smooth initial condition on the numerical convergence rate is discussed and it is shown that smoothing of initial condition helps us to achieve high-order numerical convergence rate. Moreover, Greeks (Delta and Gamma) are computed for European options and it is shown that proposed compact finite difference method is accurate for valuation of options and Greeks as well. It would be interesting to extend the proposed compact finite difference method for stochastic volatility jump-diffusion models as a future work.\\
\noindent\textbf{Acknowledgement}: Authors acknowledge the support provided by Department of Science and Technology,
India, under the grant number $SB/FTP/MS-021/2014$.
\addcontentsline{toc}{section}{References}
\bibliographystyle{unsrt}
\bibliography{references}
\end{document}